\documentclass[12pt]{article}

\usepackage{float}
\usepackage{caption}
\usepackage{subfig}
\usepackage{comment}
\usepackage{amsfonts}
\usepackage{amssymb}
\usepackage{pgfplots}
\usepackage{algorithm}
\usepackage{algorithmic}   
\usepackage{authblk}
\usepackage{amsthm}
\usepackage[margin=0.95in]{geometry}
\usepackage{graphicx, rotating}

\input{Qcircuit}

\newcommand{\naturals} {{\mathbb N}}
\newcommand{\nat} {{\mathbb N}}

\newtheorem{theorem}{Theorem}
\newtheorem{rem}{Remark}

\newtheorem{cor}{Corollary}

\newcommand\e{\varepsilon}

\newcommand{\x}{  {\hat{x}_{s_1}}     }
\newcommand{\xh}{  {\hat{x}_{s_1}}     }
\newcommand{\p}{ {\mathcal{P} }}
\newcommand{\dd}{ {\delta }}

\addtocounter{algorithm}{-1}   
\addtocounter{cor}{-1} 

\begin{document}
\title{Quantum Algorithms and Circuits for Scientific Computing}


\author{Mihir K. Bhaskar$^1$, Stuart Hadfield$^2$, Anargyros Papageorgiou$^2$, \\ and Iasonas Petras$^3$} 
\affil{$^1$ Department of Physics, Columbia University\footnote{Current address: Mihir K. Bhaskar, Dept. of Physics, Harvard University, Cambridge MA 02138.} \\
$^2$ Department of Computer Science, Columbia University\footnote{Corresponding author: stuartah@cs.columbia.edu.}\\
$^3$ Department of Computer Science, Princeton University}


\maketitle

\begin{abstract}
Quantum algorithms for scientific computing require modules implementing fundamental functions, 
such as the square root, 
the logarithm, and others. We require 
algorithms that have a well-controlled numerical error, 
that are uniformly scalable and reversible (unitary), and that can be implemented efficiently. 
We present quantum algorithms and circuits for computing the square root, the natural logarithm, 
and arbitrary fractional powers. We provide performance guarantees in terms of their worst-case accuracy and cost.
We further illustrate their performance by providing tests comparing them to 
the respective floating point implementations found in widely used numerical software.
\end{abstract}

\section{Introduction}  \label{sec:Intro}

The potential advantage of quantum computers over classical computers has generated a significant amount of 
interest in quantum computation, and has resulted in a large number of quantum algorithms not only for 
discrete problems, such as integer factorization, 
but also for computational problems in science and engineering, such as multivariate integration, path integration,
the solution of ordinary and partial differential equations, eigenvalue problems, and 
numerical linear algebra problems. A survey of such algorithms can be found in \cite{Qcontinuous}.

In solving scientific and engineering problems, classical algorithms typically use floating point arithmetic and numerical
libraries of special functions. The IEEE Standard for Floating 
Point Arithmetic (IEEE 754-2008) \cite{IEEE754-2008} ensures that such calculations are performed 
with well-defined precision. 
A similar standard is needed 
for quantum computation. 
Many quantum algorithms use the quantum circuit model of computation, typically employing a 
fixed-precision representation of numbers.
Yet there is no standard specifying how arithmetic operations between numbers (of possibly disproportionate 
magnitudes) held in registers of finite length are to be performed, and how to deal with error. 
Since registers have finite length it is not reasonable to expect to propagate all the results of 
intermediate calculations 
exactly throughout all the stages of an algorithm and intermediate approximations have to be made.
Most importantly, there are no existing libraries of quantum circuits with performance guarantees, implementing 
functions such as the square root of a number, an arbitrary fractional power of a number, the logarithm 
of a number or other similar elementary functions. 
The quantum circuits should be uniformly scalable, and at the same time make efficient use of quantum resources 
to meet physical constraints of potential quantum computing devices of the foreseeable future.

The need for such quantum circuits to be used as modules in other quantum algorithms is apparent.
For example, a recent paper deals with the solution of linear systems on a quantum computer \cite{Harrow}.
The authors present an algorithm that requires the (approximate) calculation of the reciprocal of a number followed by the calculation of 
trigonometric functions 
needed in a controlled rotation 
on the way to the final result. 
However, the paper does not give any details about how these operations are to be 
implemented. From a complexity theory point of view this may not be a complication, but certainly
there is a lot of work that is left to be done before one is able to implement the linear systems algorithm
in, say, the quantum circuit model of computation, and eventually use it
if a quantum computer becomes available.

It is worthwhile 
remarking on the direct applicability of classical algorithms to quantum computation. It is known that classical computation is subsumed by quantum computation, i.e., that for any classical algorithm, there exists a quantum algorithm which performs the same computation \cite{NC}. 
This follows from the fact that any classical algorithm (or circuit) can be implemented reversibly, in principle, but with the additional overhead of a possibly  large number of 
{\it ancilla} qubits that must be carried forward throughout the computation. 
For simple circuits consisting of the composition of basic logical operations, this overhead grows with the number of gates.
On the other hand, scientific computing algorithms are quite different. They typically  involve a large number of floating point arithmetic operations computed approximately according to rules that take into account the relative  magnitudes of the operands requiring mantissa shifting, normalization and rounding. Thus they are quite  expensive to implement reversibly because this would require many registers of large size to store all the intermediate results. Moreover, a mechanism is necessary for dealing with roundoff error and overflow which are not reversible operations. 
Hence, direct simulation on a quantum computer of classical algorithms for scientific computing that have
been implemented in floating point arithmetic quickly becomes quite complicated and prohibitively expensive.

As we indicated, we consider the quantum circuit model of computation where arithmetic operations are 
performed with fixed
precision.
We use a small number of elementary modules, or building blocks, to implement fundamental numerical functions.
Within each module the calculations are performed exactly. The results are logically truncated by selecting a desirable number of significant bits which are passed 
on as inputs to the next stage,
which is also implemented using an elementary module. We repeat this procedure until we obtain the final result. 
This way, it suffices to implement quantum mechanically a relatively small number of elementary modules, 
which can be done once, and then to combine 
them as necessary to obtain the quantum circuits implementing the different functions. 
The elementary modules carry out certain basic tasks such as shifting the bits of a number held in a 
quantum register, or counting bits, or computing 
expressions involving addition and/or multiplication of the inputs. The benefit of using only addition 
and multiplication is
that in fixed precision arithmetic the format of the input specifies exactly the format of the output, i.e., 
location of the decimal point in the result.
There exist numerous quantum circuits in the literature for addition and multiplication; 
see e.g. 
\cite{ SBN08, BASP96, ChuangCircuits,  cuccaro2004new, draper2000addition, DKRS06, rieffel2011quantum, TK05, Takahashi, VBE96}.

There are three advantages to this approach. The first is that one can derive error estimates 
by treating the elementary modules 
as black boxes and considering only the truncation error in the output of each. The second is that 
it is easy to obtain total resource estimates 
by adding the resources used by the individual modules. The third advantage is that the modular design 
allows one to modify or improve the implementation of the individual elementary modules in a transparent way.

We used this approach for Hamiltonian simulation and other numerical tasks in our paper \cite{Poisson} 
that deals with a quantum algorithm and circuit design for solving the Poisson equation. 
Individual elementary modules were combined to derive quantum circuits for Newton iteration and
to compute approximately the reciprocal of a number, and trigonometric and inverse trigonometric functions. 
We also provided performance guarantees in terms of cost and accuracy for both the individual elementary 
modules and the functions resulting by combining the modules. 
We remark that a recent paper \cite{wiebe2014quantum} also deals with quantum algorithms  for numerical analysis.

In this paper we continue this line of work of \cite{Poisson} by deriving quantum circuits which, given a number $w$ 
(represented using a finite number of bits), compute the functions $w^{1/2^i}$ for $i=1,\dots,k$, $\ln(w)$ 
(and thereby the logarithm in different bases), and $w^f$ with $f\in [0,1)$. 
For each circuit we provide cost and worst-case 
error estimates. We also illustrate the accuracy of our algorithms through a number of examples comparing 
their error with that of widely used numerical software such as Matlab. In summary, our tests show 
that using a moderate amount of resources, our algorithms compute the values of the functions 
matching the corresponding values obtained using scientific computing software (using floating point arithmetic) 
with 12 to 16 decimal digits of accuracy. 


We remark that our algorithms have applications beyond quantum computation. For instance, they can be used in 
signal processing and easily realized on an FPGA (Field Programmable Gate Array), providing superior performance with low cost. 
%
Moreover, there is resurgent 
interest in fixed-precision algorithms 
for low power/price applications such as mobile or embedded systems,
where hardware support for floating-point operations is often lacking
~\cite {bocchieri2008fixed}. 
 


We now summarize the contents of this paper. In Section \ref{sec:Algorithms}, we discuss the individual algorithms,
 providing block diagrams of the corresponding quantum circuits and pseudocode, and state their performance characteristics. In Section 
\ref{sec:NumericalResults}, we provide some numerical results illustrating the performance of our algorithms and 
comparing it to that of widely used numerical software. 
 In Section \ref{sec:implementation}, we give some remarks on the implementation of our algorithms. 
In Section \ref{sec:Discussion}, we summarize our results. 
Finally, a number of theorems about the worst-case 
error of our algorithms are given in the Appendix.


\section{Algorithms}  \label{sec:Algorithms}

We derive quantum algorithms and circuits computing
approximately $w^{1/2^i}$, $i=1,\dots,k$, $\ln(w)$ and $w^{f}$, $f\in [0,1)$, for a given input $w$.
We provide pseudocode and show how the algorithms are obtained by combining  elementary quantum circuit modules. We provide error and cost estimates.

The input of the algorithms is a fixed precision binary number. It is held in an $n$ qubit quantum register 
whose state is denoted $\ket w$ as shown in Fig. \ref{fig:InputRegister}.
The $m$ left most qubits are used to represent the integer part of the number and the remaining $n-m$  qubits 
represent its fractional part. 

\begin{figure}[H]

\centerline{
$\ket{w}=\underbrace{\ket{w^{(m-1)}} \otimes \ket{w^{(m-2)}} \otimes \cdots \otimes \ket{w^{(0)}}}_{{\rm integer\  part}} \otimes 
\underbrace{\ket{w^{(-1)}} \otimes \cdots \otimes \ket{w^{(m-n)}}}_{{\rm fractional\ part}},$
}
\caption{$n$ qubit fixed precision representation of a number $w\geq 0$ on a quantum register}
\label{fig:InputRegister}
\end{figure}

Thus 
$\ket{w} =\ket{w^{(m-1)}}\otimes \ket {w^{(m-2)}}\otimes \cdots\otimes \ket{w^{(0)}}\otimes\ket{w^{(-1)}}\otimes\cdots\otimes \ket{w^{(m-n)}}$,
where $w^{(j)}\in \{0,1\}$, $j = m-n, m-n+1, \ldots,0,\dots, m-1$ and $w = \sum_{j = m-n}^{m-1} w^{(j)}2^j$. 
Since less than $n$ bits may suffice for the representation of the input,   
a number of leftmost qubits in the register may be set to $\ket{0}$.
In general, we should have included a leading qubit to hold the sign of $w$, but since for the functions under consideration in this paper $w$ is a non-negative number we have
omitted the sign
qubit for simplicity.

Our algorithms use elementary modules that perform certain basic calculations. 
Some are used to shift the contents of registers, others are used as counters determining the position of the most significant bit of a 
number. 
An important elementary module computes expressions of the form $xy+z$ exactly in fixed precision arithmetic.

Following our convention concerning the fixed precision representation of numbers as we introduced it  in Fig. \ref{fig:InputRegister}, let  
$x$,$y$, and $z$ be represented using $n_1$-bits, of which $m_1$ bits  are used to represent the integer part. (It is not necessary to use the same number of bits to represent all
three numbers and this might useful in cases where we know that their magnitudes are significantly different.)
The expression $xy+z$ can be computed exactly as long as we allocate $2n_1 +1$  bits to hold the result. 
In this case, the rightmost $2(n_1-m_1)$ bits hold the fractional part of the result. Such computations can be 
implemented reversibly.
There are numerous quantum circuit designs in the literature implementing addition and multiplication 
\cite{VBE96, BASP96, draper2000addition, cuccaro2004new, TK05, van2005fast, DKRS06, Takahashi, portugal2006reversible, SBN08, takahashi2008fast, rieffel2011quantum, saeedi2013synthesis, kepleyquantum}. 
Therefore,  we can use them to design a quantum circuit implementing $xy+z$. In fact, we can design a quantum 
circuit template for implementing such expressions and use it to derive the actual quantum circuit for any $n_1$, $m_1$ and
values of the $x$, $y$, $z$  represented with fixed-precision.
We use Fig. \ref{fig:basicModule} below to represent such a quantum circuit.

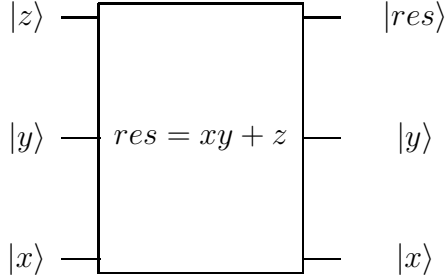
\begin{figure}[H]
\centerline{
\Qcircuit @C = 1.2 em @R = 3 em{
\lstick{\ket{z}}    &     \multigate{2}{res=xy+z}     & \qw      & &    \ket{res} \\
\lstick{\ket{y}}    &      \ghost{res=ax+b}       &      \qw     &&     \ket{y}  \\
\lstick{\ket{x}}    &      \ghost{res=ax+b}       &      \qw     &  &   \ket{x}  \\
}
}
\caption{Elementary module using fixed precision arithmetic to implement exactly $res\leftarrow xy+z$ 
for $x$, $y$, and $z$. Note that register sizes, ancilla registers, and their values are not indicated.} 
\label{fig:basicModule}
\end{figure}

Note that Fig. \ref{fig:basicModule} is an abstraction of an elementary module computing $res\leftarrow xy+z$. It is not meant to reveal or imply any of the implementation decisions including ancilla registers, saved values, and other details used for addition and multiplication. 


Any desired number $b$ of significant digits after the decimal point in the result $\ket{res}$ can be selected and passed 
on to the next stage of the computation. This corresponds to a truncation of the result to the desired accuracy.

Recall that our algorithms take as input $w$ represented by a quantum state $\ket{w}$ as shown in Fig. \ref{fig:InputRegister}. We remark that particularly for the algorithms in this
paper the integer parts of 
the inputs and outputs of the instances of the quantum circuit of  Fig. \ref{fig:basicModule} can be represented exactly using an equal number of qubits
to that used  for $\ket{w}$, i.e. $m$ qubits. 

In an earlier paper \cite{Poisson} we have cascaded such elementary modules to derive the quantum algorithm, INV, 
and the corresponding circuit computing approximately the reciprocal $1/w$ of a number $w>1$. 
The algorithm INV is based on Newton iteration, where each of its steps is implemented using the elementary module
of the form shown in Fig. \ref{fig:basicModule}. 

The algorithms of this paper compute approximations of functions which we list below and depend, to a certain 
extent, on the algorithm INV.
For this we will review INV in the next subsection. Moreover our algorithms depend upon each other.
The first algorithm we derive is SQRT computing the square root $\sqrt{w}$. This forms the basis of the algorithm 
Powerof2Roots that computes $w^{1/2^i}$, $i=1,\dots,k$. 
Powerof2Roots is used in LN, the algorithm computing $\ln(w)$. 
Without loss of generality and for brevity we only deal with the case $w>1$ in the functions computing the roots 
and the logarithm. Indeed, if $0<w<1$ one can suitably shift it to the left $\nu$ 
times to become 
$2^\nu w>1$. After obtaining the roots or the logarithm of the shifted number $2^\nu w$  
the final result for $w$ can be obtained in a straightforward way, either by shifting it to
the right in the case of the roots, or by subtracting $\nu \ln(2)$ in the case of the logarithm.

Finally we derive an algorithm computing $w^f$, $f\in [0,1)$. 
For this we distinguish two cases $w\ge 1$ and $0<w< 1$. We deal with each case separately and derive 
two algorithms FractionalPower and FractionalPower2, respectively. We remark that
when $0<w< 1$ we also  shift to the left to obtain a number greater than one, just like before. However, 
{\it undoing} the shift to get the final result is a bit complicated. For this reason we
provide all the details in FractionalPower2.
(We also note that computing the roots and the powers of $w\in\{ 0,1\}$ is
trivial and do not spend time on it since it can be accomplished with an elementary quantum circuit 
that determines if the input is zero or one. 
It is equally straightforward to compute the logarithm when $w=1$.)

In the following subsections we discuss each of our algorithms providing the main idea leading to it, its 
details along with pseudocode and quantum circuits. We have obtained theorems establishing the performance 
of our algorithms in terms of their accuracy in relation to the number of required qubits. 
We have placed these theorems in the Appendix so the reader can focus on the algorithms without having to 
consider more technical issues at the same time. At the end we provide a number of simulation test results 
illustrating the excellent performance of our algorithms and comparing them to the corresponding algorithms 
using floating point arithmetic in Matlab. Table \ref{tab:SummaryOfResults1} summarizes the algorithms in this paper,
their parameters and their error bounds.


\begin{sidewaystable}[h]
\hskip-0.5cm 
\scalebox{1}{
\footnotesize
\begin{tabular}{ | l || l | l | l | l |}
\hline
Function             	    & Requirements                                  &Algorithm and Parameters                                  	 			& Idea 									& Error \\
\hline \hline
  $1/w$               	    & $w\geq 1 $		 	   	& ${\rm INV}(w,n,m,b)$ 			 	 			& 1. Newton iteration. Calculate 						& $\leq (2 + \log_2 b)/2^b$ \\
		      	    &						& $b \ge m$			 	  	 			& $x_i=- w\hat{x}^2_{i-1} +2\hat{x}_{i-1}$, 				&  \\
		    	    &						& $s = \lceil \log_2 b \rceil	$ 	 	 			& for $i =1,2, \ldots, s$  							& \\  
			    &						&									& 2. Return $\hat x_s$							& \\
\hline
  $\sqrt w$          	    & $w\geq 1 $		 	   	&${\rm SQRT}(w,n,m,b)$		      	 			&  1. Call INV$(w,n,m,b)$	 						& $\leq \left(\frac{3}{4}\right)^{b-2m} (2+b+\log_2b)$ \\
		      	    &						&$b \geq \max\{2m,4\}$ 	 					&  2. Newton iteration. Calculate &  \\
		      	    &						&$s = \lceil \log_2 b \rceil	$	 	 			&  $y_j = \frac{1}{2} (3\hat y_{j-1} - \hat x_s \hat y^3_{j-1})$   & \\
		        	    &						&					 	 			&  for $j = 1,2, \ldots, s$							& \\
			    &						&									&  3. Return $\hat y_s$							& \\
\hline
  $w^{1/2^i}$    	    & $w\geq 1 $		 	   	& ${\rm Powerof2Roots}(  w,k,  n,m,b)$    			&  1. $z_1 = \textrm{SQRT}(w,n,m,b)$ 				& $\leq 2\left(\frac{3}{4}\right)^{b-2m} (2+b+\log_2b)$ \\
$i = 1,2, \ldots k$ 	    &						&  $b \geq \max\{2m,4\}$	 			&  2. Call SQRT() repeatedly, i.e., 						&  \\
		      	    &						&$s = \lceil \log_2 b \rceil	$ for	 	 			&   $z_i = \textrm{SQRT}(z_{i-1},m+b,m,b)$,				&\\
		      	    &						& each call of SQRT()		 	 			&  for $i = 1,2, \ldots k$							& \\
			    &						&									& 3. Return $\{z_i\}$							& \\
\hline
  $\ln w$  	   	    & $w\geq 1 $		 	   	&${\rm LN}(	w,n,m,\ell)$	         		 			&  1. $w_p = w\cdot 2^{1-p}$ 						& $\leq \left( \frac{3}{4}\right)^{5\ell/2} \left(m + \frac{32}{9} +2 \left(\frac{32}{9} + \frac{n}{\ln 2}\right)^3 \right)$ \\
		 	    &						& $b = \max\{5\ell,25\}$			 			&  2. Call 								 	&  \\
		      	    &						&  $\ell \geq \lceil \log_2 8n \rceil$ 		  		&  $\textrm{PowerOf2Roots}(w_p,\ell,n,1,b)$ 	 		& \\
		      	    &						&$r \approx \ln 2$, with $b$ bits accuracy 			&   and let $\hat t_p$ be the $\frac{1}{2^\ell}$th root of $w_p$		&\\
		      	    &						& $p = \lceil\log_2 w\rceil$			 			&  3. Approx. $\ln \hat t_p$ with $\hat y_p$, i.e., 			& \\
		      	    &						&				 					&  the first two terms of 					  		& \\
		      	    &						&				 	 				&  its power series expansion. 				 		& \\
		      	    &						&					 			    	&  4. Return $z_p = 2^\ell \hat y_p + (p-1) r$ 	 		& \\
\hline
  $w^f$     		    & $w\geq 1 $		 	   	& ${\rm FractionalPower}(w,f,n,m,n_f,\ell)$		 	&  1. Calculate 								& $\leq \left( \frac{1}{2}\right)^{\ell -1}$ \\
			    & $f \in [0,1]$				& $b = \max\{n,n_f, \lceil 5(\ell, 2m,\ln n_f) \rceil\}$   	& $\hat w_i = \textrm{PowerOf2Roots}(w,n_f,n,m,b)$		& \\
   			    &$f$ is $n_f$ bits long			& $\ell \in \nat$ determines the error						&  for $i = 1,2, \ldots, n_f $.					 	&  \\
		      	    &						& 		&  2. Return $\Pi_{i\in \mathcal P} \hat w_i$ 				& \\
		      	    &						& 					 	 			&   where $\mathcal P = \{1\leq i \leq n_f: f_i = 1\}$			&\\
\hline
$w^f$		    & $0 \leq w < 1$				&${\rm FractionalPower2}(w,f,n,m,n_f,\ell)$		 	& 1. Compute $w^\prime \geq 1$ 				& $\leq \frac{1}{2^{\ell -3}}$ \\
			    & $f \in [0,1]$				&$b = \max\{n,n_f, \lceil 2\ell + 6m + 2\ln n_f \rceil, 40\}$	& by left shifting $w$								& \\
			    & $f$ is $n_f$ bits			& $\ell \in \nat$ determines the error			& 2. Call 									& \\
			    &						&					& $\textrm{FractionalPower}(w^\prime,f,n,m,n_f,\ell)$		& \\
			    &						&									& 3. \textit{Undo} the initial shift	 of $w$				& \\
                &						&									& using right shifts, FractionalPower,  								& \\
                &                   &                                   & and INV, and return                                                   &\\
\hline
\end{tabular}
}
\caption{Summary of Algorithms. All parameters are polynomial in $n$ and $b$ and so is the cost of all algorithms.}
\label{tab:SummaryOfResults1}
\end{sidewaystable}

\clearpage

\subsection{Reciprocal}

An algorithm computing the reciprocal of a number is shown in Section 4.2 and Theorem B.1 of \cite{Poisson}. The algorithm is based on Newton iteration. Below we provide a slight modification of that algorithm. 

Recall that $w$ is represented with $n$ bits of which the first $m$ correspond to its
integer part.
Algorithm \ref{alg:inv} INV below approximates the reciprocal of a number $w\geq 1$, applying Newton iteration to the function $f(x) = \frac{1}{w} -x$.
This yields a sequence of numbers $x_i$  according to the iteration
\begin{equation}
\label{eq:NI.INV}
x_{i} =g_1(x_{i-1}) := -w\hat x_{i-1}^2 +2\hat x_{i-1},
\end{equation}
$i = 1,2,\ldots ,s$. 
Observe that the expression above can be computed using two applications of a quantum circuit of the type shown in Fig. \ref{fig:basicModule}. 
The initial approximation $\hat x_0 =2^{-p}$, with 
$2^p > w \geq 2^{p-1}$. The number of iterations $s$ is specified in Algorithm \ref{alg:inv} INV. 
Note that $x_0<  1/w$ and the iteration converges to $1/w$ from below, i.e., $\hat x_i \le 1/w$.
Within each iterative step the arithmetic operations are performed in fixed precision and $x_i$ is computed exactly. We truncate $x_i$ to $b\ge n$ bits after the decimal point to obtain $\hat x_i$ and pass it on as input to the next iterative step. Each iterative step is implemented using an elementary module of the form given in Fig. \ref{fig:basicModule} that requires only addition and multiplication. 
 The final approximation error is 
$$
| \hat{x}_s - \frac{1}{w}| \leq \frac{2+ \log_2 b}{2^b}.
$$
For the derivation of this error bound see Corollary \ref{cor0} in the Appendix. 
We remark that although the iteration function (\ref{eq:NI.INV}) is well known in the literature \cite[Ex. 5-1]{traub1982iterative}, 
an important property of Algorithm \ref{alg:inv} INV is the fixed-precision implementation of Newton iteration for a specific initial approximation and a prescribed number of steps, so that the error bound of Corollary \ref{cor0} is satisfied. 

Turning to the cost, we iterate $O(\log_2 b)$ times  and as we mentioned each $x_i$ is computed exactly. Therefore, each iterative step requires $O(n+b)$ qubits and a  number of
quantum operations for implementing addition and  multiplication that is a low degree 
polynomial in $n+b$. The cost to obtain the initial approximation is 
relatively minor 
when compared to the overall cost of the multiplications and additions used in the algorithm. 

\begin{algorithm}
\caption{INV($w$, $n$, $m$, $b$)}  
\label{alg:inv}
\begin{algorithmic}[1]
\REQUIRE $w\geq 1$, held in an $n$ qubit register, of which the first $m$ qubits are reserved for its integer part.
\REQUIRE $b \in \nat$, $b \geq m$. We perform fixed precision arithmetic and results are truncated to $b$ bits of accuracy after the decimal point. 
\IF{$w=1$} 
\RETURN 1
\ENDIF
\STATE $\hat{x}_0 \leftarrow 2^{-p}$, where $p\in\nat$ such that $2^p > w \geq 2^{p-1}$
\STATE $s \leftarrow \lceil \log_2 b \rceil$
\FOR{$i=1$ to $s$}
\STATE $x_i \leftarrow -w\hat{x}_{i-1}^2 + 2\hat{x}_{i-1} $
\STATE $\hat{x}_i \leftarrow x_i$ truncated to $b$ bits after the decimal point  
\ENDFOR
\RETURN $\hat{x}_{s}$
\end{algorithmic}
\end{algorithm}


\subsection{Square Root}

Computing approximately the square root $\sqrt{w}$, $w\ge 1$, can also be approached as a zero finding problem and one can 
apply to it Newton iteration.
However, the selection of the function whose zero is $\sqrt{w}$ has to be done carefully so that the resulting iterative steps
are easy to implement and analyze in terms of error and cost. Not all choices are equally good. 
For example, $f(x)= x^2-w$,
although well known in the literature \cite[Ex. 5-1]{traub1982iterative},
is not a particularly good choice. The resulting iteration is $x_{i+1} = x_i - (x_i^2 - w)/ (2x_i)$, $i=0,1,\dots$, which requires a division using an algorithm
such as  Algorithm~\ref{alg:inv}~INV at each iterative step. The division also requires circuits keeping 
track of the position of the decimal point in its result, because its location is not fixed but depends on 
the values $w$ and $x_i$. Since the result of a division may not be represented exactly using an a priori chosen 
fixed number of bits,
approximations are needed within each iterative step. This introduces error and overly complicates the analysis of the overall algorithm 
approximating $\sqrt{w}$ compared to an algorithm requiring only 
multiplication and addition in each iterative step. All these complications are avoided in our algorithm.

\begin{figure}[H]

$$\qquad\qquad \begin{array}{c}
\Qcircuit @C = 1 em @R = 2 em{
\lstick{\ket{\hat y_0}}  & \qw & \qw & \qw & \qw & \qw & \qw & \qw & \qw &\qw & \qw & \qw & \qw & \qw & \multigate{1}{g_2} & \qw & & \ket{\hat y_1} & & \dots & & \multigate{1}{g_2} & \qw & & \ket{\hat y_{s_2}}  & & & & 
\\
\lstick{\ket{\hat x_0}} & \multigate{1}{g_1} & \qw & & \ket{\hat x_1} & & \dots & & \multigate{1}{g_1} & \qw & & \ket{\hat x_{s_1}} & & & \ghost{g_2} & \qw & & \ket{\hat x_{s_1}} & & \dots & & \ghost{g_2} & \qw & & \ket{\hat x_{s_1}} & & & & 
\\
\lstick{\ket{w}} & \ghost{g_1} & \qw & & \ket{w} & & \dots & & \ghost{g_1} & \qw & & \ket{w}  & & & \qw & \qw & \qw & \qw & \qw& \qw &\qw & \qw & \qw & & \ket{w} \\
& & & & & \raisebox{-.5 em}{$s_1\ \text{iterations}$} & & & & & & & & & & & & \raisebox{-.5 em}{$s_2\ \text{iterations}$}
\gategroup{2}{2}{3}{9}{2.1 em}{_\}} \gategroup{1}{15}{3}{22}{2.1 em}{_\}}
}
\end{array}$$

\caption{Block diagram of the overall circuit computing $\sqrt{w}$. Two stages of Newton's iteration using the functions $g_1$ and $g_2$ are applied $s_1$ and $s_2$ times respectively. The first stage outputs $\hat x_{s_1} \approx \frac{1}{w}$, which is then used by the second stage to compute $\hat y_{s_2}\approx \tfrac 1{\sqrt{\hat x_{s_1}}} \approx \sqrt{w}$. } 

\label{fig-SQRToverall}
\end{figure}
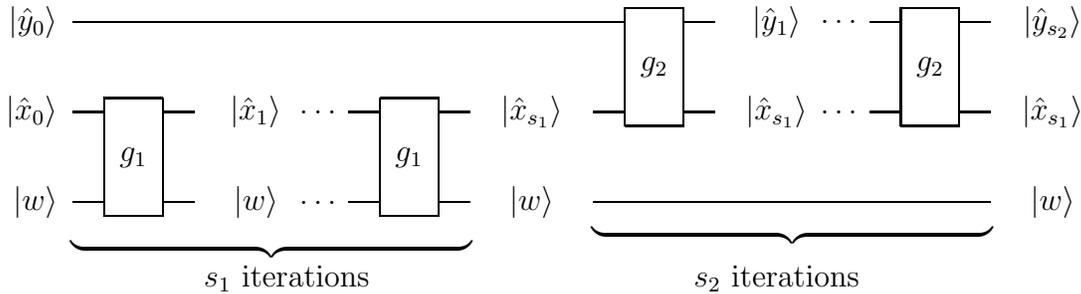

Each of the steps of the algorithm we present below can be implemented using only multiplication and addition in 
fixed precision arithmetic as in Fig. \ref{fig:basicModule}.
This is accomplished by first approximating $1/w$ (applying iteration $g_1$ of equation (\ref{eq:NI.INV})) using steps essentially identical to those in Algorithm \ref{alg:inv} INV. 
Then we apply Newton iteration again to a suitably chosen function to approximate its zero and this yields the approximation of $\sqrt{w}$. 
In particular, in Algorithm \ref{alg:inv} INV we set $s=s_1=s_2$ and  first we approximate $1/w$ by $\hat x_s$ which has a fixed precision representation with $b$ bits after the decimal point
(steps 4--10 of Algorithm \ref{alg:inv} INV). 
Then applying the Newton method to $f(y) = \tfrac 1{y^2} - \tfrac 1{w}$ we obtain the
iteration 
\begin{equation}
\label{eq:NI.SQRT}
y_j = g_2(y_{j-1}) := \frac{1}{2} (3\hat y_{j-1} - \hat x_s \hat y^3_{j-1}),
\end{equation}
where $j = 1,2,\dots,s$. Each $y_i$ is computed exactly and then truncated  to $b$ bits after the decimal point to obtain
$\hat y_i$, which is passed on as input to the next iterative step. See Algorithm \ref{alg:sqrt} SQRT for the values of the parameters and
other details. Steps 4 and 10 of the algorithm compute initial approximations for the two Newton iterations, the first computing 
the reciprocal and the second shown in (\ref{eq:NI.SQRT}).
They are implemented 
using the quantum circuits of  Fig. \ref{fig-InitState1} and Fig. \ref{fig-InitState2}, respectively.
A block diagram of the overall circuit of the algorithm computing $\sqrt{w}$ is shown in Fig. \ref{fig-SQRToverall}.

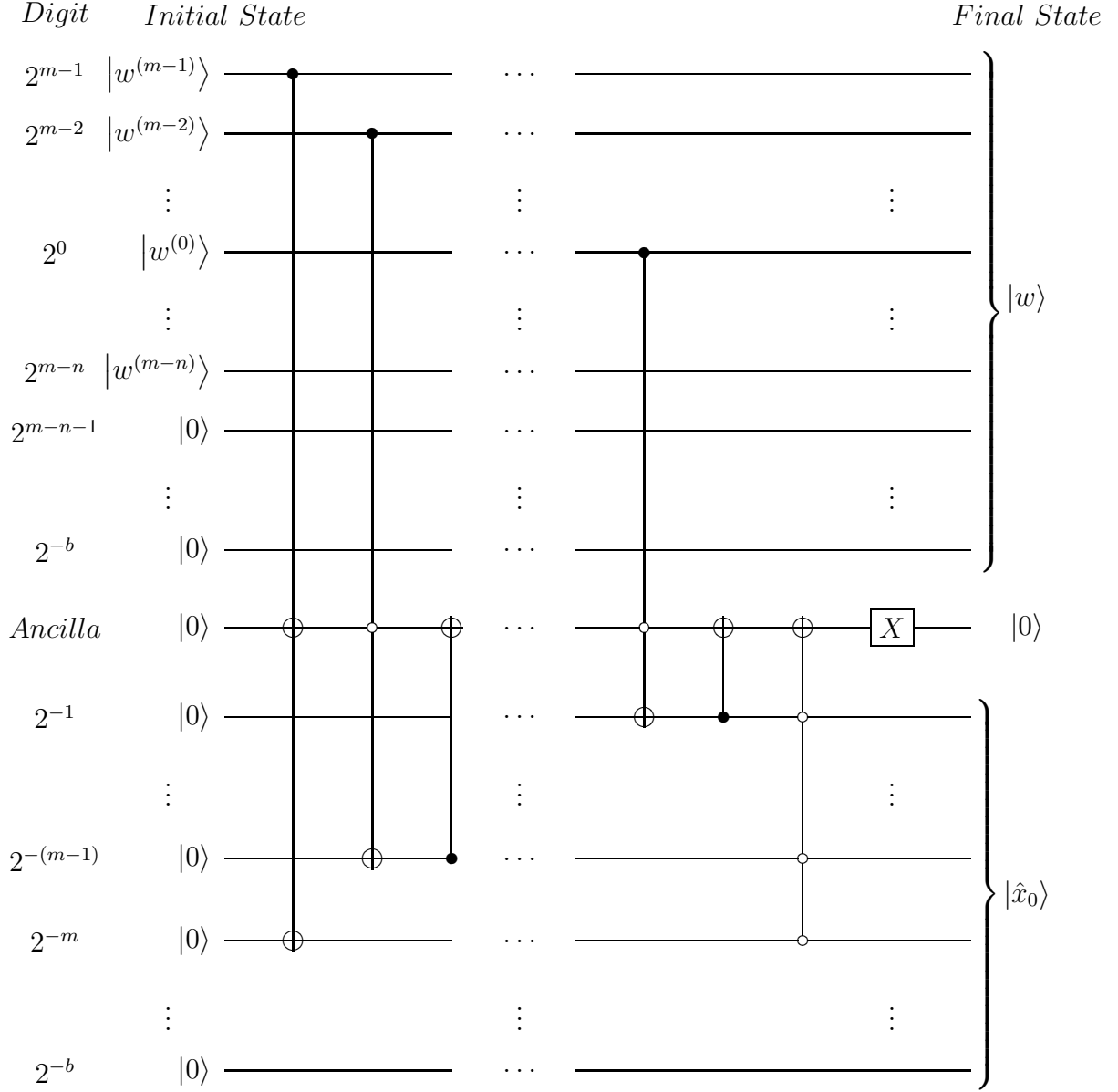
\begin{figure}[H]
\centerline{
\Qcircuit @C = 2.0 em @R = 2.1 em{
Digit & & & Initial\ State & & & & & & & & & & & Final\ State \\ 
2^{m-1} & & & \lstick{\ket{w^{(m-1)}}} & \ctrl{13} & \qw & \qw & \dots & & \qw & \qw & \qw & \qw & \qw \\
2^{m-2} & & & \lstick{\ket{w^{(m-2)}}} & \qw & \ctrl{8} & \qw & \dots & & \qw & \qw & \qw & \qw & \qw \\
& & \vdots & & & & & \vdots & & & & & \vdots  \\
2^0 & & & \lstick{\ket{w^{(0)}}} & \qw & \qw & \qw & \dots & & \ctrl{6} & \qw & \qw & \qw & \qw & \raisebox{-4 em}{$\ket{w}$} \\
& & \vdots & & & & & \vdots & & & & & \vdots \\
2^{m-n} & & & \lstick{\ket{w^{(m-n)}}} & \qw & \qw & \qw & \dots & & \qw & \qw & \qw & \qw & \qw \\
2^{m-n-1} & & & \lstick{\ket{0}} & \qw & \qw & \qw & \dots & & \qw & \qw & \qw & \qw & \qw \\
& & \vdots & & & & & \vdots & & & & & \vdots \\
2^{-b} & & & \lstick{\ket{0}} & \qw & \qw & \qw & \dots & & \qw & \qw & \qw & \qw & \qw \\
Ancilla & & & \lstick{\ket{0}} & \targ & \ctrlo{3} & \targ & \dots & & \ctrlo{1} & \targ & \targ & \gate{X} & \qw & \ket{0} \\
2^{-1} & & & \lstick{\ket{0}} & \qw  & \qw & \qw & \dots & & \targ & \ctrl{-1} & \ctrlo{-1} & \qw & \qw \\
& & \vdots & & & & & \vdots & & & & & \vdots \\
2^{-(m-1)} & & & \lstick{\ket{0}} & \qw & \targ & \ctrl{-3} & \dots & & \qw & \qw & \ctrlo{-2} & \qw & \qw \\
2^{-m} & & & \lstick{\ket{0}} & \targ & \qw & \qw & \dots & & \qw & \qw & \ctrlo{-1} & \qw & \qw & \raisebox{3.5 em}{$\ket{\hat x_0}$} \\
& & \vdots & & & & & \vdots & & & & & \vdots \\
2^{-b} & & & \lstick{\ket{0}} & \qw & \qw & \qw & \dots & & \qw & \qw  & \qw & \qw & \qw
\gategroup{2}{14}{10}{14}{1.5 em}{\}} \gategroup{12}{14}{17}{14}{1.2 em}{\}}
}
}
\caption{A quantum circuit computing the initial state $\ket{\hat x_0} = \ket{2^{-p}}$, for $\ket{w}$ given by $n$ bits of which $m$ are for its integer part, where $p\in \nat$ and $2^p > w \geq 2^{p-1}$. Here $w^{(m-1)},\dots, w^{(0)}$ label the $m$ integral bits of $w$ and similarly $w^{(-1)},\dots, w^{(m-n)}$ label the fractional bits. We have taken $b\ge n-m$.
This circuit is used in step 4 of Algorithm 1 SQRT.} 

\label{fig-InitState1}
\end{figure}

For $w\ge 1$ represented with $n$ bits of which the first $m$ correspond its integer part, Algorithm  \ref{alg:sqrt} SQRT 
computes $\sqrt{w}$ by $\hat y_s$ 
in fixed precision with $b\ge \max\{ 2m, 4\}$ bits after its decimal point and we have
$$
|\hat{y}_{s}-\sqrt{w}|  
\leq \left( \frac{3}{4} \right)^{b-2m}  \left( 2+ b + \log_2 b  \right).
$$
The proof can be found in Theorem \ref{thm1} in the Appendix.

\begin{figure}[H]
\centerline{
\Qcircuit @C = 1.2 em @R = 2 em{
Digit & & & & Initial\ State & & & & & & & & & & & & & & & & & & Final\ State \\
2^{-1} & & & & & & \lstick{\ket{\hat{x}^{(-1)}_{s}}} & \ctrl{13} & \ctrlo{1} & \qw & \qw & \qw & \qw & \qw & \dots & & \qw & \qw & \qw & \qw & \qw & \qw \\
2^{-2} & & & & & &  \lstick{\ket{\hat{x}^{(-2)}_{s}}} & \qw & \ctrl{12} & \qw & \qw & \qw & \qw & \qw & \dots & & \qw &  \qw & \qw & \qw & \qw & \qw \\
2^{-3} & & & & & &  \lstick{\ket{\hat{x}^{(-3)}_{s}}} & \qw & \qw & \qw & \ctrl{6} & \ctrlo{1} & \qw & \qw & \dots & & \qw & \qw & \qw & \qw & \qw & \qw \\
2^{-4} & & & & & & \lstick{\ket{\hat{x}^{(-4)}_{s}}} & \qw & \qw & \qw & \qw & \ctrl{5} & \qw &\qw & \dots & & \qw & \qw & \qw & \qw & \qw & \qw \\
2^{-5} & & & & & & \lstick{\ket{\hat{x}^{(-5)}_{s}}} & \qw & \qw & \qw & \qw & \qw & \qw & \ctrl{4} & \dots & & \qw & \qw & \qw & \qw & \qw & \qw & & \raisebox{3 em}{$\ket{\hat{x}_{s}}$} \\
& & & & & \vdots & & & & & & & & & \vdots & & & & & & \vdots \\
2^{-b+1} & & & & & & \lstick{\ket{\hat{x}^{(-b+1)}_{s}}} & \qw & \qw & \qw & \qw & \qw & \qw & \qw & \dots & & \ctrl{2} & \ctrlo{1} & \qw & \qw & \qw & \qw \\
2^{-b} & & & & & & \lstick{\ket{\hat{x}^{(-b)}_{s}}} & \qw & \qw & \qw & \qw & \qw & \qw & \qw & \dots & & \qw & \ctrl{1} & \qw & \qw & \qw & \qw \\
Ancilla & & & & & & \lstick{\ket{0}} &  \qw & \qw & \targ & \ctrlo{4} & \ctrlo{4} & \targ & \ctrlo{3} & \dots & & \ctrlo{1} & \ctrlo{1}  & \targ & \targ & \gate{X} & \qw & & \ket{0} \\
2^{\frac{b}{2} - 1} & & & & & & \lstick{\ket{0}} & \qw & \qw & \qw & \qw & \qw & \qw & \qw & \dots & & \targ & \targ & \ctrl{-1} & \ctrlo{-1} & \qw & \qw \\
& & & & & \vdots & & & & & & & & & \vdots & & & & & & \vdots \\
2^{2} & & & & & & \lstick{\ket{0}} & \qw & \qw & \qw &\qw & \qw & \qw & \targ & \dots & & \qw & \qw & \qw & \ctrlo{-2} & \qw & \qw \\
2^{1} & & & & & & \lstick{\ket{0}} & \qw & \qw & \qw & \targ & \targ & \ctrl{-4} & \qw & \dots & & \qw & \qw & \qw & \ctrlo{-1} & \qw & \qw \\
2^{0} & & & & & & \lstick{\ket{0}} & \targ & \targ & \ctrl{-5} & \qw & \qw & \qw & \qw & \dots & & \qw & \qw & \qw & \ctrlo{-1} & \qw & \qw & & \raisebox{4 em}{$\ket{\hat y_0}$} \\
2^{-1} & & & & & & \lstick{\ket{0}} & \qw & \qw & \qw & \qw & \qw & \qw & \qw & \dots && \qw & \qw & \qw  & \qw & \qw & \qw \\
& & & & & \vdots & & & & & & & & &  \vdots & & & & & & \vdots \\
2^{-b} & & & & & & \lstick{\ket{0}} & \qw & \qw & \qw & \qw & \qw & \qw & \qw & \dots & & \qw & \qw & \qw & \qw & \qw & \qw
\gategroup{2}{22}{9}{22}{1.5 em}{\}} \gategroup{11}{22}{18}{22}{1.5 em}{\}}
}
}
\caption{A quantum circuit computing the state $\ket{\hat y_0} = \ket{2^{\lfloor \frac{q-1}{2} \rfloor}}$, for $0 < \hat{x}_s < 1 $ given by $b$ bits, where $q\in \nat$ and $2^{1-q} > \hat{x}_s \geq 2^{-q}$. Here again $\hat{x}_s^{(-1)},\dots, \hat{x}_s^{(-b)}$  are the $b$ fractional bits. This circuit is used in step 10 of Algorithm 1 SQRT. 
The case for even $b$ is shown here and for odd $b$ a similar circuit follows.}
\label{fig-InitState2}
\end{figure}
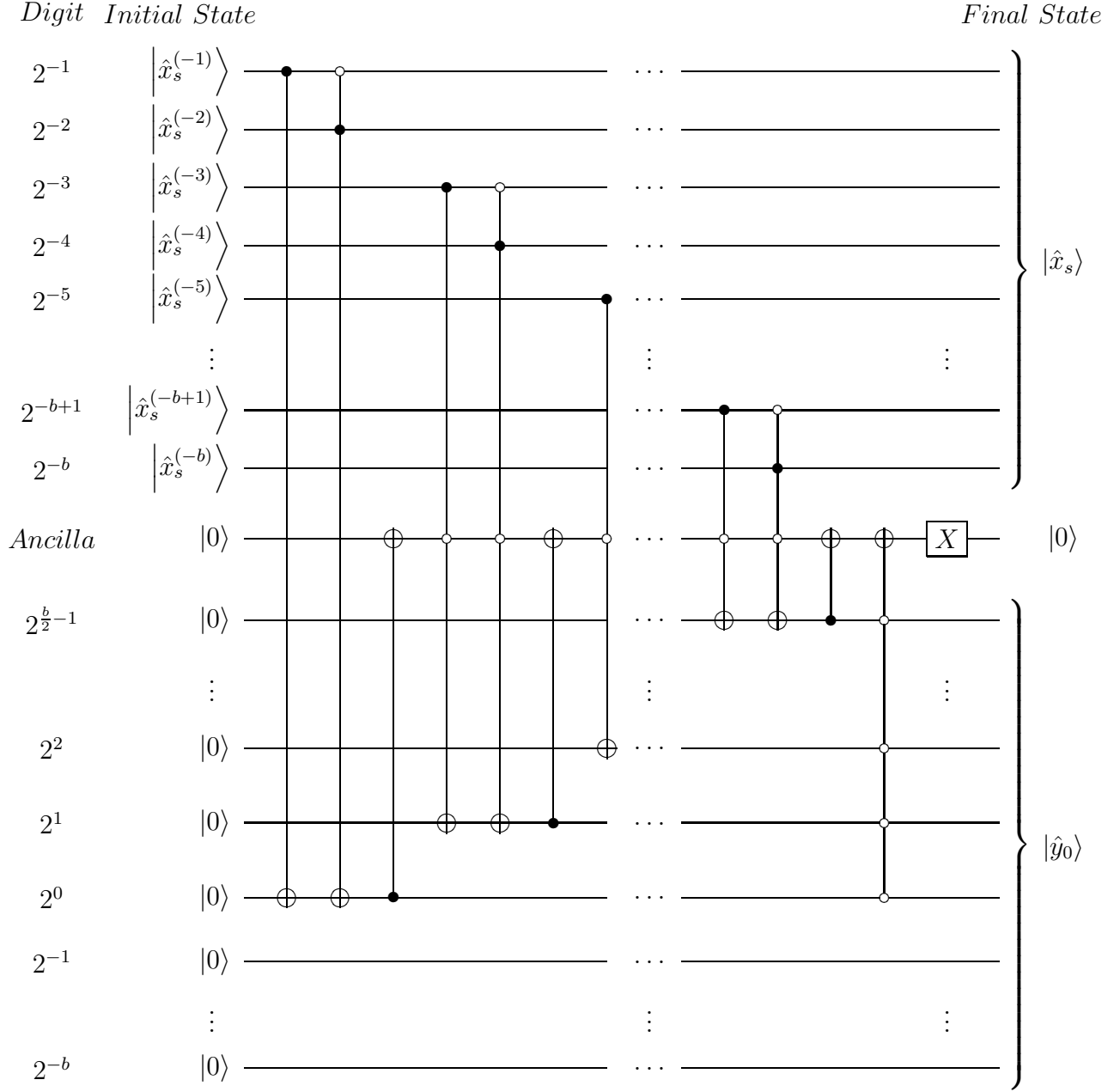 

Turning to the cost of the quantum circuit in Fig. \ref{fig-SQRToverall} implementing Algorithm \ref{alg:sqrt} SQRT, we set the number of iterative steps of each of the two iterations to be
$s=s_1=s_2=O(\log_2 b)$. 
Observe that each of the iterations $g_1$  and $g_2$ can be computed using at most three applications of a quantum circuit of the type shown in Fig. \ref{fig:basicModule}. 
Therefore, each iterative step requires $O(n+b)$ qubits and a  number of
quantum operations for implementing addition and  multiplication that is a low degree 
polynomial in $n+b$. 
Observe that the cost to obtain the initial approximations 
is again relatively minor compared to the overall cost of the additions and multiplications in the algorithm.

\begin{algorithm}
\caption{SQRT($w$, $n$, $m$, $b$)}
\label{alg:sqrt}
\begin{algorithmic}[1]
\REQUIRE $w\geq 1$, held in an $n$ qubit register, of which the first $m$ qubits are reserved for its integer part.
\REQUIRE $b\geq \max\{2m,4\}$. Results are truncated to $b$ bits of accuracy after the decimal point.
\IF{$w=1$} 
\RETURN 1
\ENDIF
\STATE $\hat{x}_0 \leftarrow 2^{-p}$, where $p\in\nat$ such that $2^p > w \geq 2^{p-1}$
\STATE $s \leftarrow \lceil \log_2 b \rceil$
\FOR{$i=1$ to $s$}
\STATE $x_i \leftarrow -w\hat{x}_{i-1}^2 + 2\hat{x}_{i-1} $
\STATE $\hat{x}_i \leftarrow x_i$ truncated to $b$ bits after the decimal point  
\ENDFOR
\STATE $\hat{y}_0 \leftarrow 2^{\lfloor (q-1)/2 \rfloor}$, where $q\in \nat$ such that $2^{1-q} > \hat{x}_{s}  \geq 2^{-q}$
\FOR{$j=1$ to $s$}
\STATE $y_j \leftarrow \frac12(3\hat{y}_{j-1} - \hat{x}_{s} \hat{y}^3_{j-1} )$
\STATE $\hat{y}_j \leftarrow y_j$ truncated to $b$ bits after the decimal point  
\ENDFOR
\RETURN $\hat{y}_{s}$
\end{algorithmic}
\end{algorithm}

\subsection{$2^k$-Root}

Obtaining the roots, $w^{1/2^i}$, $i=1,\dots,k$, $k\in\naturals$, is straightforward. It is accomplished by calling Algorithm \ref{alg:sqrt} SQRT iteratively $k$ times, since
$w^{1/2^i}= \sqrt{ w^{1/2^{i-1}}}$, $i=1,\dots,k$. In particular, 
Algorithm \ref{alg:2krt} PowerOf2Roots calculates approximations of $w^{1/2^i}$, for $i=1,2, \ldots, k$. 
The circuit implementing this algorithm consists of $k$ repetitions of the circuit in Fig. \ref{fig-SQRToverall} approximating the square root.
The results are truncated to $b\ge \max\{ 2m, 4\}$ bits after the decimal point before passed on to the next stage.
The algorithm produces $k$ numbers $\hat z_i$, $i=1,\dots,k$. We have
$$|\hat{z}_{i}-w^{1/{2^i}}| \leq 2  \left( \frac{3}{4} \right)^{b-2m}  \left( 2+ b +  \log_2 b  \right),$$ 
$i=1,\dots,k$. 
The proof can be found in Theorem \ref{thm3} in the Appendix.

\begin{algorithm}
\caption{PowerOf2Roots($w$, $k$, $n$, $m$, $b$)}
\label{alg:2krt}
\begin{algorithmic}[1]
\REQUIRE $w\geq 1$, held in an $n$ qubit register, of which the first $m$ qubits are reserved for its integer part.
\REQUIRE $k\geq 1$ an integer. The algorithm returns approximations of $w^\frac{1}{2^i}$, $i=1,\dots,k$.
\REQUIRE $b\geq \max\{2m,4\}$. Results are truncated to $b$ bits of accuracy after the decimal point.
\STATE $\hat z_1 \leftarrow$ SQRT($w$, $n$, $m$, $b$). Recall that SQRT returns a number with a fractional part $b$ bits long. The integer part of of $\hat z_1$ is represented by $m$ bits.
\FOR{$i=2$ to $k$}
\STATE $\hat z_i \leftarrow$ SQRT($\hat{z}_{i-1}$, $m+b$, $m$, $b$). Note that $\hat z_1$ and the $\hat z_i$ are held in registers of size $m+b$ bits of which the $b$ bits are for the fractional part.
\ENDFOR
\RETURN $\hat z_1$,$\hat z_2$,\dots,$\hat z_k$
\end{algorithmic}
\end{algorithm}

Algorithm \ref{alg:2krt} PowerOf2Roots uses $k$ calls to Algorithm \ref{alg:sqrt} SQRT. Hence it requires $k\log b\cdot O(n+b)$ qubits and $k\log b\cdot p(n+b)$ quantum operations, where $p$ is a low-degree polynomial depending on the specific implementation of the circuit of Fig. \ref{fig:basicModule}.

\subsection{Logarithm}
To the best of our knowledge, the method presented in this section is entirely new. Let us first introduce the idea leading to the algorithm approximating $\ln(w)$, $w> 1$; the case $w=1$ is trivial. First we shift $w$ to the left, if necessary, to obtain the number $2^{-\nu} w \in [1,2)$.
It suffices to approximate $\ln (2^{-\nu} w)$ since 
$\ln( w) = \ln (2^{-\nu}  w) +\nu \ln 2$ and we can precompute $\ln 2$ up to any desirable number of bits. 

We use the following observation. When one takes the $2^\ell$-root of  a number that belongs to the interval  $(1,2)$  the fractional part $\delta$ of the result is 
roughly speaking proportional to $2^{-\ell}$, i.e, its is quite small for relatively  large $\ell$. 
Therefore, for $1+\dd = [2^{-\nu}w]^{1/2^\ell}$ we use the power series expansion for the logarithm to approximate $\ln(1+\dd)$ by $\dd- \dd^2/2$, with any desired accuracy 
since $\delta$ can be
made arbitrarily small by appropriately selecting $\ell$. 
Then the approximation of the logarithm follows from 
\begin{equation}
\label{eq:undolog}
\ln(w)  \approx \nu\ln 2 + 2^\ell(\dd- \frac {\dd^2}2).
\end{equation}

In particular, Algorithm \ref{alg:ln} LN approximates $\ln(w)$ for $w\ge 1$ represented by $n$ bits of which the first $m$ are used for its integer part. 
(The trivial case $w=1$ is dealt with first.)
In step 7, $p-1$ is the value of $\nu$, i.e., $\nu=  p-1$, where  $2^p> w\ge 2^{p-1}$, $p\in\naturals$. 
We compute 
$w_p=2^{1-p}w \in [1,2)$ using a right shift of $w$. We explain how to implement the shift operation below.
Then, $\hat t_p$, an approximation of $w_p^{1/(2\ell)}$, is calculated using Algorithm \ref{alg:2krt} PowerOf2Roots, 
where $\ell$ is a parameter that determines the error. 
Note that $\hat t_p$ can be written as $1+\dd$, for $0<\dd< 1$. We provide precise bounds for $\dd$ in Theorem \ref{thm3} in the Appendix. 
Next, an approximation $\hat y_p$ of $\ln \hat t_p=\ln(1+\dd)$ is calculated, using the first two terms of the power series for the logarithm as we explained above. 
The whole procedure yields an estimate of $(\ln w_p)/2^\ell $. 
Finally the algorithm in steps 16 -- 20  uses equation (\ref{eq:undolog}), with $\nu=p-1$, to
derive $z_p+(p-1) r$ as an approximation to $\ln w$. 
All intermediate calculation results are truncated to $b$ bits after the decimal point. The value of $b$ is determined by the value of $\ell$ in step 1 of
the algorithm. Note that the precision of the algorithm grows with $\ell$, which is a user selected parameter that must satisfy $\ell\ge\lceil \log_2 8n\rceil$.

The block diagram of the overall circuit implementing Algorithm \ref{alg:ln} LN can be found in Fig. \ref{fig-LogOverall}. The first module is a right shift 
operation that calculates $w_p$. 
During a presentation of these results, we were asked about the implementation of the shift operation. There are a number of ways to implement it. 
One is to implement the shift by a multiplication of $w$ by a suitable negative power of two. 
This is convenient since we 
have elementary quantum circuits for multiplication.  In Fig.~\ref{fig-RShift} we show a circuit computing  the necessary power of two, which we denote by $x$ in that figure. In particular, 
for $w\ge 2$ we set $x=2^{1-p}$, where $p-1=\lfloor \log_2 w\rfloor \ge 1$.
For $1\le w <2$ we set $x=1$. 
Thus $m$ bits are needed for the representation of $x$,
with the first bit  $x^{(0)}$ denoting its integer part and all the remaining bits $x^{(-1)}, \dots,
x^{-(m-1)}$  denoting its fractional part.
We implement the shift of $w$ in terms of multiplication between $w$ and $x$,
i.e., $w_p=w\, x$. Since $w$ and $x$ are held in registers of known size and we use fixed precision representation we know a priori the position of the decimal point in their product $w_p$. Moreover, since $w_p\in [1,2-2^{1-n}]$ ($w$ is an $n$ bit number) we have that $n$ bits (qubits), of which $1$ is used for its integer part, suffice to hold $w_p$ exactly.

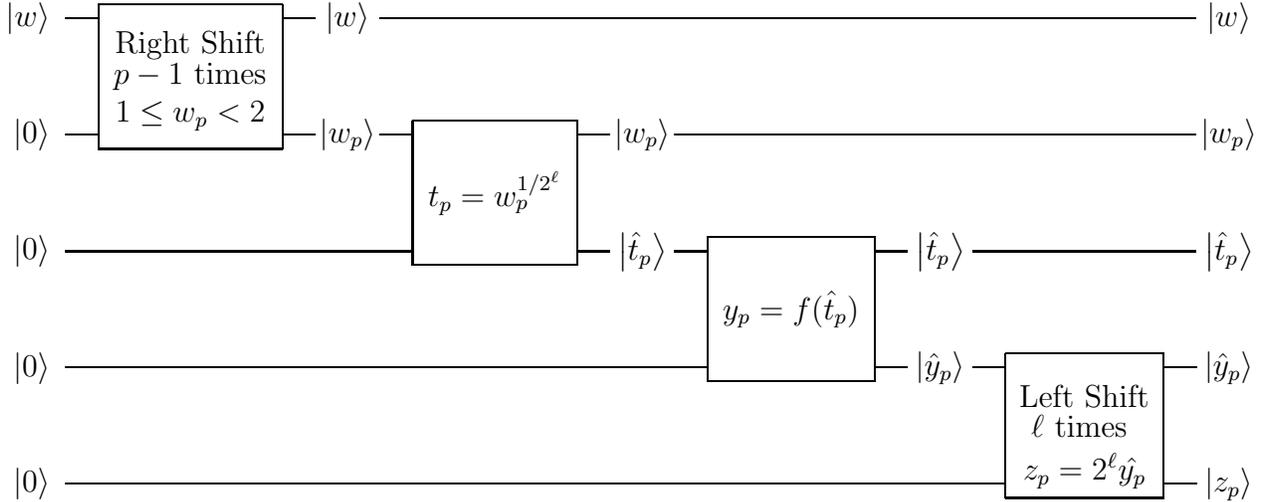
\begin{figure}[H]
\centerline{
\Qcircuit @C = 1.05 em @R = 2.85 em{
\lstick{\ket{w}} & \multigate{1}{p-1\ \text{times}} & \qw & \ket{w} & & \qw & \qw & \qw & \qw & \qw & \qw & \qw & \qw & \qw & \qw & \ket{w} \\
\lstick{\ket{0}} & \ghost{p-1\ \text{times}} & \qw & \ket{w_p} & & \multigate{1}{t_p=w_p^{1/2^\ell}} & \qw & \ket{w_p} & & \qw & \qw & \qw & \qw & \qw \raisebox{-17.5em}{\text{Left Shift}} & \qw & \ket{w_p} \\
\lstick{\ket{0}} & \qw \raisebox{9 em}{$1 \leq w_p<2$} & \qw & \qw & \qw & \ghost{t_p=w_p^{1/2^\ell}} & \qw & \ket{\hat{t}_p}  & & \multigate{1}{y_p=f(\hat{t}_p)} & \qw & \ket{\hat{t}_p} & & \qw \raisebox{-15em}{$z_p=2^\ell\hat{y_p}$} & \qw & \ket{\hat{t}_p} \\
\lstick{\ket{0}} & \qw \raisebox{21em}{\text{Right Shift}}& \qw & \qw & \qw & \qw & \qw & \qw & \qw & \ghost{y_p=f(\hat{t}_p)} & \qw & \ket{\hat{y}_p} & & \multigate{1}{\ \ell\ \text{times}\ \ } & \qw & \ket{\hat{y}_p} \\
\lstick{\ket{0}} & \qw & \qw & \qw & \qw & \qw & \qw & \qw & \qw & \qw & \qw & \qw & \qw & \ghost{\ \ell\ \text{times}\ \ } & \qw & \ket{z_p} 
}
}
\caption{Overall circuit schematic for approximating $\ln w$.  
The state $\ket{\ell}$ holds the required input parameter $\ell$ which determines the accuracy of Algorithm 3 LN. 
The gate $f(\hat{t}_p)$ outputs $\hat{y}_p =  (\hat{t}_p - 1) - \frac{1}{2} (\hat{t}_p - 1)^2$. 
Once $z_p$ is obtained the approximation of $\ln w$ is computed by the expression
 $z_p + (p-1)r$,  where $r$ approximates $\ln 2$ with high accuracy and $p-1$ is obtained from a quantum circuit as the one shown in Fig.~\ref{fig-compute_p}.}

\label{fig-LogOverall}
\end{figure}

The next module is the circuit for the PowerOf2Roots algorithm and calculates $\hat t_p$.
The third module calculates $\hat y_p$ and is comprised of modules that perform subtraction and multiplication in fixed precision. 
The fourth module of the circuit performs a series of left shift operations. 
Observe that $\ell$ is an accuracy parameter whose value is set once at the very beginning and does not change during the execution of the algorithm.
Thus the left shift $\ell$ times is, in general, much easier to implement than the right shift of $w$  at the beginning of the algorithm and we omit the details.

In its last step Algorithm \ref{alg:ln} LN computes the expression $\hat z:=z_p+(p-1)r$ which is the estimate of $\ln w$. The value of $p-1$ in this expression is obtained using
the quantum circuit of Fig. \ref{fig-compute_p}. The error of the algorithm satisfies
$$ |  \hat{z}  - \ln w | \leq   \left(\frac{3}{4}\right)^{5\ell/2} \left( m+ \frac{32}{9} + 2\left(\frac{32}{9} + \frac{n}{\ln 2} \right)^3 \right).$$
For the proof details see Theorem \ref{thm3} in the Appendix.

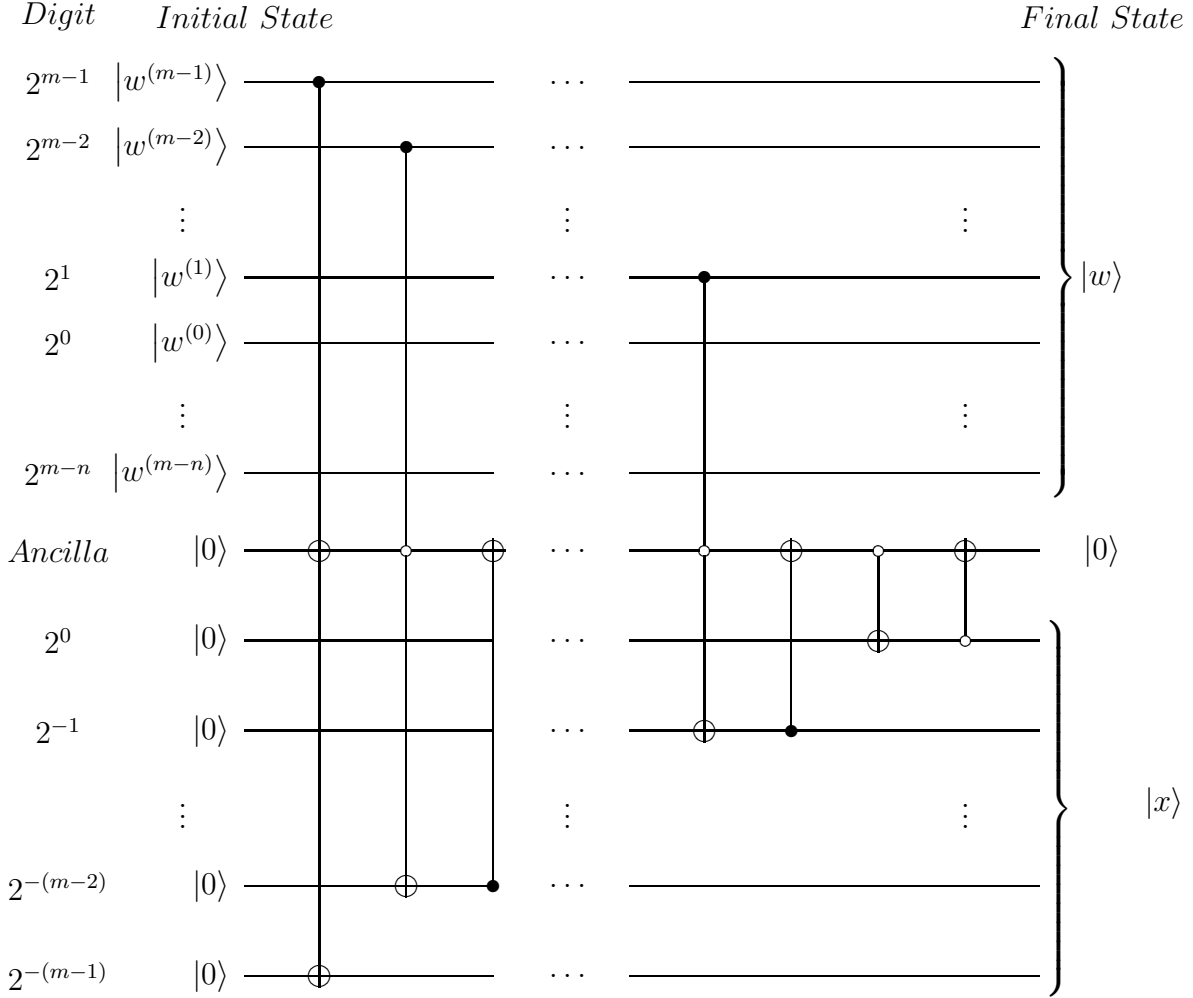
\begin{figure}[H]
\centerline{
\Qcircuit @C = 2.0 em @R = 2.1 em{
Digit & & & Initial\ State & & & & & & & & & & & Final\ State \\ 
2^{m-1} & & & \lstick{\ket{w^{(m-1)}}} & \ctrl{12} & \qw & \qw & \dots & & \qw & \qw & \qw & \qw &  \qw \\
2^{m-2} & & & \lstick{\ket{w^{(m-2)}}} & \qw & \ctrl{6} & \qw & \dots & & \qw & \qw & \qw & \qw & \qw\\
& & \vdots & & & & & \vdots & & & & & \vdots   \\
2^1 & & & \lstick{ \ket{w^{(1)}} } & \qw & \qw & \qw & \dots & & \ctrl{4} & \qw & \qw & \qw & \qw &  \raisebox{0 em}{$\ket{w}$} \\
2^{0} & & & \lstick{  \ket{w^{(0)}}    } & \qw & \qw & \qw & \dots & & \qw & \qw  & \qw & \qw & \qw \\
& & \vdots & & & & & \vdots & & & & & \vdots \\
2^{m-n} & & & \lstick{\ket{w^{(m-n)}}} & \qw & \qw & \qw & \dots & & \qw & \qw & \qw & \qw & \qw  \\
Ancilla & & & \lstick{\ket{0}} & \targ & \ctrlo{4} & \targ & \dots & & \ctrlo{2} & \targ  & \ctrlo{1} & \targ & \qw & \ket{0} \\
2^{0} & & & \lstick{\ket{0}} & \qw & \qw & \qw & \dots & & \qw & \qw  &  \targ & \ctrlo{-1} & \qw \\
2^{-1} & & & \lstick{\ket{0}} & \qw  & \qw & \qw & \dots & & \targ & \ctrl{-2} &  \qw & \qw & \qw \\
& & \vdots & & & & & \vdots & & & & & \vdots & &  &  \raisebox{0.6 em}{$\ket{ x}$}  \\
2^{-(m-2)} & & & \lstick{\ket{0}} & \qw & \targ & \ctrl{-4} & \dots & & \qw & \qw &  \qw & \qw & \qw \\
2^{-(m-1)} & & & \lstick{\ket{0}} & \targ & \qw & \qw & \dots & & \qw & \qw &  \qw & \qw & \qw 
\gategroup{2}{14}{8}{14}{1.5 em}{\}} \gategroup{10}{14}{14}{14}{1.2 em}{\}}
}
}
\caption{For $w\ge 1 $ this quantum circuit computes $\ket{x}$ where $x$ is an $m$ bit number $x\in [2^{1-m},1]$. 
Here $w^{(m-1)},\dots, w^{(0)}$ label the $m$ integral bits of $w$ and similarly $w^{(-1)},\dots, w^{(m-n)}$ label the fractional bits.
 For $w\ge 2$ we set $x=2^{1-p}$, where $p-1=\lfloor \log_2 w\rfloor \ge 1$.
For $1\le w <2$ we set $x=1$. 
Thus $m$ bits are needed for the representation of $x$,
with the first bit  $x^{(0)}$ denoting its integer part and all the remaining bits $x^{(-1)}, \dots,
x^{-(m-1)}$  denoting its fractional part.
This circuit is used in steps 6  -- 10 of Algorithm~3~LN to  derive $x=2^{1-p}$  so one can implement the shift of $w$ in terms of multiplication between $w$ and $x$,
i.e., $w_p=w\, x$. Since $w_p\in [1,2-2^{1-n}]$ we have that $n$ qubits of which $1$ is used for the integer part suffice to hold $w_p$ exactly.}
\label{fig-RShift}
\end{figure}

Considering the cost of the quantum circuit in Fig. \ref{fig-LogOverall} implementing Algorithm \ref{alg:ln} LN, we have that the cost of computing the initial and the last shifts 
(first and fourth modules in Fig.~\ref{fig-LogOverall}),
as well as the cost of the arithmetic expression in the third module in Fig.~\ref{fig-LogOverall}, the cost of computing $p-1$ (see the circuits in Fig.~\ref{fig-RShift} and  Fig.~\ref{fig-compute_p} ) 
and the cost of the expression $z_p+(p-1)r$,
are each relatively minor compared to the other costs of the algorithm. 
The algorithm calls Algorithm~\ref{alg:2krt} PowerOf2Roots with input parameter $k:= \ell$, which requires 
 $O( \ell (n+b)\log b)$ qubits and $\ell \log b \cdot p(n+b)$ quantum operations, as explained in the previous subsection. Observe that the expression in Step 17 of Algorithm \ref{alg:ln} LN can be computed using a constant number of applications of a quantum circuit of the type shown in Fig. \ref{fig:basicModule}. Hence, the overall cost of Algorithm \ref{alg:ln} LN is $O( \ell (n+b)\log b)$ qubits and requires a number of quantum operations proportional to $\ell \log b \cdot p(n+b)$.

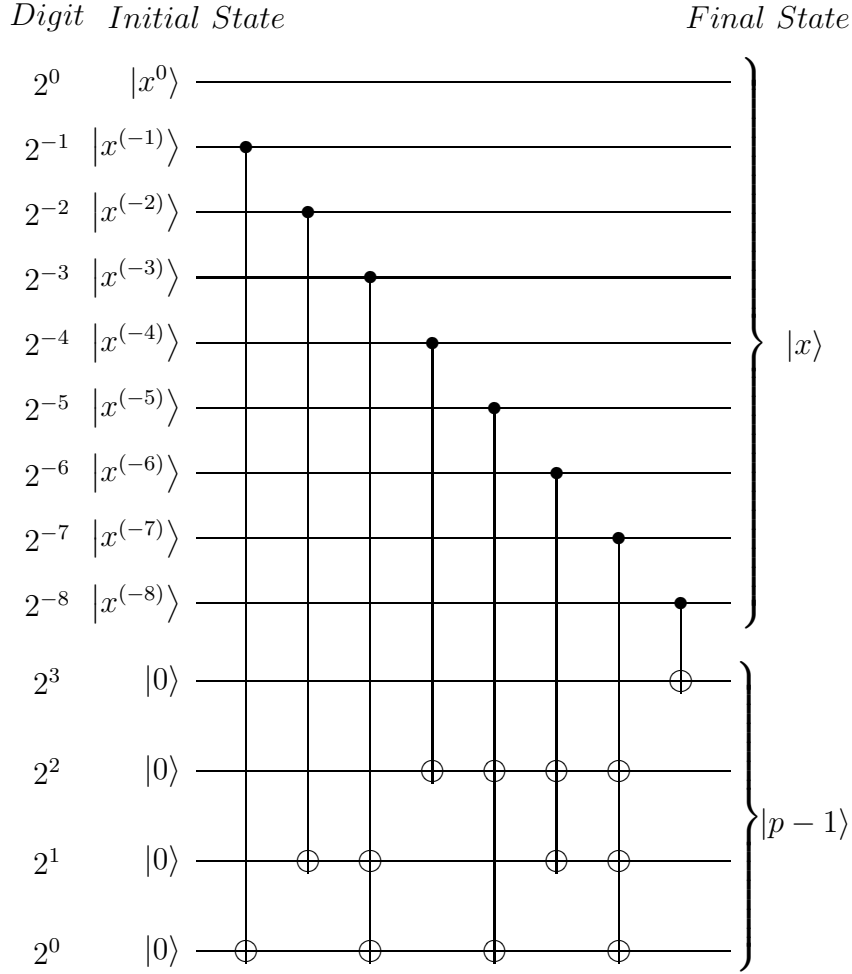
\begin{figure}[H]

\centerline{
\Qcircuit @C = 1.2 em @R = 2.1 em{
Digit & & & & Initial\ State & & & & & & & & & & Final\ State \\
2^{0} & & & & \lstick{\ket{x^0}}  & \qw & \qw & \qw & \qw & \qw & \qw & \qw & \qw & \qw \\
2^{-1} & & & & \lstick{\ket{x^{(-1)}}}  & \ctrl{11} & \qw  & \qw & \qw & \qw & \qw & \qw & \qw & \qw \\
2^{-2} & & & & \lstick{\ket{x^{(-2)}}} & \qw & \ctrl{9}  & \qw & \qw & \qw & \qw & \qw & \qw & \qw \\
2^{-3} & & & & \lstick{\ket{x^{(-3)}}} & \qw & \qw & \ctrl{9}  & \qw & \qw & \qw & \qw & \qw & \qw \\
2^{-4} & & & & \lstick{\ket{x^{(-4)}}} & \qw & \qw & \qw & \ctrl{6}  & \qw & \qw & \qw & \qw & \qw & & \raisebox{-1 em}{$\ket{x}$} \\
2^{-5} & & & & \lstick{\ket{x^{(-5)}}} & \qw & \qw & \qw & \qw & \ctrl{7} & \qw & \qw & \qw & \qw  \\
2^{-6} & & & & \lstick{\ket{x^{(-6)}}} & \qw & \qw & \qw & \qw & \qw & \ctrl{5} & \qw & \qw & \qw \\
2^{-7} & & & & \lstick{\ket{x^{(-7)}}} & \qw & \qw & \qw & \qw & \qw & \qw & \ctrl{5} & \qw & \qw \\
2^{-8} & & & & \lstick{\ket{x^{(-8)}}} & \qw & \qw & \qw & \qw & \qw & \qw & \qw & \ctrl{1} & \qw \\
2^3 & & & & \lstick{\ket{0}} & \qw & \qw & \qw & \qw & \qw & \qw & \qw & \targ & \qw & & \raisebox{-10 em}{$\ket{p-1}$}\\
2^2 & & & & \lstick{\ket{0}} & \qw & \qw & \qw  & \targ & \targ & \targ  & \targ & \qw & \qw \\
2^1 & & & & \lstick{\ket{0}}  & \qw & \targ  & \targ & \qw & \qw & \targ & \targ & \qw & \qw \\
2^0 & & & & \lstick{\ket{0}} & \targ  & \qw  & \targ  & \qw  & \targ & \qw  & \targ & \qw & \qw \gategroup{2}{14}{10}{14}{1.5 em}{\}} \gategroup{11}{14}{14}{14}{1.2 em}{\}}
}
}
\caption{Example of a quantum circuit computing $p-1\ge 0$ required in the last step of Algorithm 3 {LN}. The input to this circuit is
the state $\ket{x}$ computed in Fig.~\ref{fig-RShift} where $x=2^{-(p-1)}$. Recall that $m$ bits are used to store $x$, and clearly $\lceil \log_2 m \rceil$ bits suffice to store $p-1$ exactly.
In this example, $m = 9$. It is straightforward to generalize this circuit to an arbitrary number $m$. }

\label{fig-compute_p}
\end{figure}

\begin{algorithm}
\caption{LN($w$, $n$, $m$, $\ell$)}
\label{alg:ln}
\begin{algorithmic}[1]
\REQUIRE $w\geq 1$, held in an $n$ qubit register, of which the first $m$ qubits are reserved for its integer part.
\REQUIRE 
$\ell \geq \lceil \log_2 8n\rceil$ is a parameter upon which the error will depend and we will use to determine the number of bits $b$ after the decimal points in which arithmetic will be performed.
\STATE $b\leftarrow \max\{5\ell,25\}$. Results are truncated to $b$ bits of accuracy after the decimal point in the intermediate calculations.
\STATE $r$ $\leftarrow$ $\ln 2$ with $b$ bits of accuracy, i.e. $|r-\ln 2|\le 2^{-b}$. An approximation of $\ln 2$  can be precomputed with sufficiently many bits of accuracy and stored in a register, from which we take the first $b$ bits.
\IF{$w=1$} 
\RETURN 0
\ENDIF
\STATE Let $p\in \nat$ is such that $2^p > w \geq 2^{p-1}$
\IF {$p-1 =0$}
\STATE $w_p \leftarrow w$. In this case $w=w_p\in [1,2-2^{1-n}]$.
\ELSE 
\STATE $w_p \leftarrow w2^{1-p}$. Note that $w_p\in [1,2-2^{1-n}]$ for $w\ge 2$. The number of bits (qubits) used for $w_p$ is $n$ of which $1$ is for its integer part as explained in the text and the caption of  Fig.~\ref{fig-RShift}.
\STATE $x_p \leftarrow w_p - 1$. This is the fractional part of $w_p$. 
\ENDIF
\IF{$x_p=0$} 
\STATE $z_p \leftarrow 0$
\ELSE 
\STATE $\hat{t}_p \leftarrow $PowerOf2Roots($w_p$, $\ell$, $n$, $1$, $b$)$[\ell]$. The function $PowerOf2Roots$ returns a list of numbers and we take the last element, the $1/2^\ell$th root. Note that in this particular case  $1\leq \hat{t}_p < 2$. 
\STATE $\hat{y}_p \leftarrow (\hat{t}_p-1) - \frac12 (\hat{t}_p -1)^2$, computed to $b$ bits of accuracy after the decimal point. Note that $\hat{t}_p = 1 + \delta$ with $\delta \in (0,1)$, and we approximate $\ln(1+\delta)$ by $\delta-\frac12 \delta^2$, the first two terms of its power series expansion.
\STATE $z_p \leftarrow 2^{\ell}\hat{y}_p$. This corresponds to a logical right shift of the decimal point. 
\ENDIF
\RETURN $z_p+(p-1)r$ 
\end{algorithmic}
\end{algorithm}

\subsection{Fractional Power}

Another application of interest is the approximation of fractional powers of the form $w^f$, where $w\geq 1$ and
$f\in [0,1]$ a number whose binary form is $n_f$ bits long. The main idea is to calculate appropriate powers of 
the form $w^{1/2^i}$ according to the value of the bits of $f$ and multiply the results.
Hence initially, the algorithm PowersOf2Roots is used to derive a list of approximations $\hat w_i$ to the powers 
$w_i = w^{1/2^i}$, for $i = 1,2, \ldots ,n_f$.
The final result that approximates $w^{f}$
is $\Pi_{i\in \mathcal P} \hat w_i$, where $\mathcal P = \{1\leq i \leq n_f: f_i = 1\}$ and $f_i$ denotes the $i$th bit of the number $f$.
The process is described in detail in Algorithm \ref{alg:fracPower} Fractional Power.

For $w>1$ the value $\hat z$ returned by the algorithm satisfies
$$
|\hat{z} - w^{f}|  \leq \left( \frac{1}{2} \right)^{\ell  - 1  } , 
$$
where $w$ is represented by $n$ bits of which $m$ are used for its integer part, $n_f$ is number of bits in the representation of the exponent $f$, $\ell\in\naturals$ is a user selected parameter determining the accuracy of the result. The results of all intermediate steps are truncated to 
$b\ge\max\{ n,n_f,[5(\ell+2m+\ln n_f)],40\}$ bits before passing them  on to the next step. The proof can be found in Theorem \ref{thm4} in the Appendix.

The algorithm can be extended to calculate $w^{p/q}$, $w>1$, where the exponent is a rational number $p/q \in [0,1]$. 
First $f$, an $n_f$ bit number, is calculated such that it approximates $p/q$ within $n_f$ bits of accuracy, 
namely $|f-\frac{p}{q}| \leq 2^{-n_f}$.
Then $f$ is used as the exponent in the parameters of Algorithm \ref{alg:fracPower} FractionalPower to 
get an approximation of $w^f$, which in turn is an approximation of $w^{p/q}$.
The value $\hat z$ returned by the algorithm satisfies 
$$
|\hat{z} - w^{p/q}|  \leq     \left( \frac{1}{2} \right)^{\ell  - 1  }  + \frac{w\ln w}{2^{n_f}}. 
$$
The proof can be found in Corollary \ref{cor1} in the Appendix.

\begin{rem}
For example, for $p/q=1/3$, one can use 
Algorithm \ref{alg:inv} INV to produce an approximation $f$ of $1/3$ and pass that to the algorithm.
In such a case, the approximation error $|w^{p/q} - w^f| \leq 2^{-n_f} w \ln w$, obtained using the mean value theorem for the function $w^x$,  appears as the last term of the equation above. For the details see corollaries \ref{cor1} and  \ref{cor2} in the Appendix. 
\end{rem}

When $w\in(0,1)$ we can shift it appropriately to obtain a number greater than one to which we can apply   Algorithm \ref{alg:fracPower} FractionalPower.
However when approximating the fractional power {\it undoing} the  initial shift to obtain an estimate of $w^f$ is a bit more involved than it is for the previously considered
functions. For this reason we provide the details in Algorithm~\ref{alg:fracPower2}FractionalPower2.
The algorithm first computes $k$ such that $2^k w \geq 1 > 2^{k-1}w$; see Fig. \ref{fig-ShiftInteger}. Using $k$ the algorithm shifts $w$ to the left to obtain $x:=2^k w$. The next step is
to approximate $x^f$. Observe that $x^f=2^{kf} w^f$. Therefore to undo the initial shift we have to divide by $2^{kf}= 2^c  2^{\{ c\} }$, where $c$ denotes the integer part of $kf$ and
$\{c\}$ denotes its fractional part. Dividing by $2^c$ is straightforward and is accomplished using shifts. 
Dividing by $2^{ \{ c\} }$ is accomplished by first approximating  $2^{ \{ c\} }$  and then
multiplying by its approximate reciprocal. See Algorithm \ref{alg:fracPower2} FractionalPower2 for all the details. 

The value $\hat t$ the algorithm returns as an approximation of $w^f$, $w\in(0,1)$, satisfies
$$
|\hat{t} - w^{f}| \leq   \frac{1}{2^{\ell - 3 }},
$$
where $\ell$ is a user-defined parameter just like before.
The proof can be found in Theorem \ref{thm5} in the Appendix.

\begin{algorithm}
\caption{FractionalPower($w$, $f$, $n$, $m$, $n_f$, $\ell$)}
\label{alg:fracPower}
\begin{algorithmic}[1]
\REQUIRE $w\geq 1$, held in an $n$ qubit register, of which the first $m$ qubits are reserved for its integer part.
\REQUIRE $\ell \in \nat$ is a parameter upon which the error will depend. We use it to determine the number of bits $b$ after the decimal points in which arithmetic will be performed.
\REQUIRE $1\geq f \geq 0$ is a binary string corresponding to a fractional number given with $n_f$
bits of accuracy after the decimal point. 
The algorithm returns an approximation of $w^f$.
\STATE $b \leftarrow \max \{ n,n_f, \lceil 5(\ell + 2m +\ln n_f) \rceil,40\}$. Results are truncated to $b$ bits of accuracy after the decimal point.
\IF{$f=1$}
   \RETURN w
\ENDIF
\IF{$f=0$}
   \RETURN 1
\ENDIF
\STATE $\{ \hat{w}_i  \} \leftarrow $PowerOf2Roots($w$, $n_f$, $n$, $m$, $b$). The function returns a list of numbers  $ \hat{w}_i$  approximating $w_i =  w^{\frac{1}{2^i}}$, $i=1,\dots,n_f$.
\STATE $\hat{z} \leftarrow 1$
\FOR{$i=1$ to $n_f$}
   \IF{the $i$th bit of $f$ is $1$}
       \STATE $z \leftarrow \hat{z} \hat{w}_i$
       \STATE $\hat{z} \leftarrow z$ truncated to $b$ bits after the decimal point
   \ENDIF
\ENDFOR
\RETURN $\hat{z}$
\end{algorithmic}
\end{algorithm}

Just like before we can approximate $w^p/q$, $w\in(0,1)$, and a rational exponent $p/q\in[0,1]$, by first approximating $p/q$ and then calling 
Algorithm \ref{alg:fracPower2}  FractionalPower2. The value $\hat t$ the algorithm returns satisfies
$$
|\hat{t} - w^{p/q}|  \leq     \left( \frac{1}{2} \right)^{\ell  - 2  }  + \frac{w\ln w}{2^{n_f}}.
$$
See Corollary \ref{cor2} in the Appendix.

\begin{algorithm}
\caption{FractionalPower2($w$, $f$, $n$, $m$, $n_f$, $\ell$)}
\label{alg:fracPower2}
\begin{algorithmic}[1]
\REQUIRE $0 \leq w< 1$, held in an $n$ qubit register, of which the first $m$ qubits are reserved for its integer part.  (For $w\geq 1$, use Algorithm \ref{alg:fracPower}.)
\REQUIRE $\ell \in \nat$ is a parameter upon which the error will depend. We use it to determine the number of bits $b$ after the decimal points in which arithmetic will be performed.
\REQUIRE $1\geq f \geq 0$ specified to $n_f$ 
bits of accuracy. 
The algorithm returns an approximation of $w^f$.
\STATE $b \leftarrow \max \{ n,n_f, \lceil 2\ell + 6m +2\ln n_f \rceil,40\}$. Results are truncated to $b$ bits of accuracy after the decimal point.
\IF{$w=0$}
   \RETURN 0
\ENDIF
\IF{$f=1$}
   \RETURN w
\ENDIF
\IF{$f=0$}
   \RETURN 1
\ENDIF
 \STATE $x \leftarrow 2^k w$, where $k$ is a positive integer such that  $2^k w \geq 1 > 2^{k-1} w$. This corresponds to a logical right shift of the decimal point. An example of the quantum circuit computing $k$ is given in Fig. \ref{fig-ShiftInteger}.
 \STATE $c \leftarrow kf$  
 \STATE $\hat{z} \leftarrow  $FractionalPower($x$, $f$, $n$, $m$, $n_f$, $\ell$). This statement computes an approximation of $z=x^f$.
 \STATE $\hat{y} \leftarrow  $FractionalPower($2$, ${\{c\}}$, $n$, $m$, $n_f$, $\ell$). This statement computes an approximation of $y=2^{\{c\}}$, where $\{c\} = c -  \lfloor c \rfloor$ (for $c\geq 0$) denotes the fractional part of $c$. Since we use fixed precision arithmetic, the integer and fractional parts of numbers are readily available.
  \STATE $\hat{s} \leftarrow$ INV$(\hat{y}, n, 1,2\ell)$. This statement computes  an approximation of $s=\hat{y}^{-1}$. 
\STATE $v \leftarrow 2^{- \lfloor c \rfloor} \hat{z}$. This corresponds to a logical left shift of the decimal point.
\STATE $t \leftarrow v \hat{s}$
\STATE $\hat{t} \leftarrow t$, truncated to $b$ bits after the decimal point.
\RETURN $\hat{t}$
\end{algorithmic}
\end{algorithm}

We now address the cost of our Algorithms computing $w^f$. First consider Algorithm \ref{alg:fracPower} FractionalPower, which calls Algorithm \ref{alg:2krt} PowerOf2Roots with input parameters $k:= n_f$ and $b:=\max \{ n,n_f, \lceil 5(\ell + 2m +\ln n_f) \rceil,40\}$, which requires $O( n_f  (n+b)\log b)$ qubits and of order $n_f \log b \cdot p(n+b)$ quantum operations, as explained in the previous subsections. At most $n_f$ multiplications are then required, using a quantum circuit of the type shown in Fig. \ref{fig:basicModule}. Hence,  Algorithm \ref{alg:fracPower} FractionalPower  requires a number of qubits and a number of quantum operations that is 
a low-degree polynomial in $n,n_f$, and $\ell$, respectively.

\begin{figure}[H]

\centerline{
\Qcircuit @C = 1.2 em @R = 2.1 em{
Digit & & & & Initial\ State & & & & & & & & & & & & &  &  & & Final\ State \\
2^{m-1} & & & & \lstick{\ket{0}} & \qw & \qw & \qw & \qw & \qw & \qw & \qw & \qw & \qw & \qw & \qw & \qw & \qw & \qw & \qw & \qw \\
& & & \vdots \\
2^{0} & & & & \lstick{\ket{0}} & \qw & \qw & \qw & \qw & \qw & \qw & \qw & \qw & \qw &\qw & \qw & \qw & \qw & \qw & \qw & \qw \\
2^{-1} & & & & \lstick{\ket{w^{(-1)}}} & \ctrl{12} & \qw & \qw & \qw & \qw & \qw & \qw & \qw & \qw & \qw & \qw & \qw & \qw & \qw & \qw & \qw \\
2^{-2} & & & & \lstick{\ket{w^{(-2)}}} & \qw & \ctrl{7} & \qw & \qw & \qw & \qw  & \qw & \qw & \qw & \qw & \qw & \qw & \qw & \qw & \qw & \qw \\
2^{-3} & & & & \lstick{\ket{w^{(-3)}}} & \qw & \qw  & \qw & \ctrl{6} & \qw & \qw & \qw & \qw & \qw & \qw & \qw & \qw & \qw & \qw & \qw & \qw & \raisebox{-1 em}{$\ket{w}$} \\
2^{-4} & & & & \lstick{\ket{w^{(-4)}}} & \qw & \qw & \qw & \qw & \qw & \ctrl{5} & \qw & \qw & \qw & \qw & \qw & \qw & \qw & \qw & \qw & \qw \\
2^{-5} & & & & \lstick{\ket{w^{(-5)}}} & \qw & \qw & \qw & \qw  & \qw & \qw &\qw & \ctrl{4} & \qw  & \qw & \qw & \qw & \qw & \qw & \qw & \qw \\
2^{-6} & & & & \lstick{\ket{w^{(-6)}}} & \qw & \qw & \qw & \qw & \qw & \qw & \qw & \qw & \qw & \ctrl{3} & \qw & \qw & \qw & \qw & \qw & \qw \\
2^{-7} & & & & \lstick{\ket{w^{(-7)}}} & \qw & \qw & \qw & \qw & \qw & \qw & \qw & \qw & \qw & \qw & \qw & \ctrl{2} & \qw & \qw & \qw & \qw \\
2^{-8} & & & & \lstick{\ket{w^{(-8)}}} & \qw & \qw & \qw & \qw & \qw & \qw & \qw & \qw & \qw & \qw & \qw & \qw & \qw & \ctrl{1} & \qw & \qw \\
Ancilla & & & & \lstick{\ket{0}} & \targ & \ctrlo{3} & \targ & \ctrlo{4} & \targ & \ctrlo{2} & \targ & \ctrlo{4} & \targ &\ctrlo{3} & \targ & \ctrlo{4} & \targ & \ctrlo{1} & \targ & \gate{X} & \ket{0} \\
2^3 & & & & \lstick{\ket{0}} & \qw & \qw & \qw & \qw & \qw & \qw & \qw & \qw & \qw & \qw & \qw & \qw & \qw & \targ & \ctrl{-1} & \qw & \raisebox{-10 em}{$\ket{k}$}\\
2^2 & & & & \lstick{\ket{0}} & \qw & \qw & \qw & \qw & \qw & \targ & \ctrl{-2} & \targ & \ctrl{-2} & \targ & \ctrl{-2} & \targ & \ctrl{-2} & \qw & \qw & \qw \\
2^1 & & & & \lstick{\ket{0}} & \qw & \targ & \ctrl{-3} & \targ & \ctrl{-3} & \qw & \qw & \qw & \qw & \targ & \ctrl{-1} & \targ & \ctrl{-1} & \qw & \qw & \qw  \\
2^0 & & & & \lstick{\ket{0}} & \targ & \qw & \qw & \targ & \ctrl{-1} & \qw & \qw & \targ & \ctrl{-2} & \qw & \qw  & \targ & \ctrl{-1} & \qw & \qw & \qw \gategroup{2}{21}{12}{21}{1.5 em}{\}} \gategroup{14}{21}{17}{21}{1.2 em}{\}}
}
}
\caption{Example of quantum circuit computing the positive integer $k$ such that $2^k w \geq 1 > 2^{k-1} w$,
for $0 < w < 1$ having $n-m$ significant bits after the decimal point. 
In this example, $n-m = 8$. 
It is straightforward to generalize this circuit for arbitrary number of bits. 
This circuit is used in Algorithm 5 FractionalPower2. 
}

\label{fig-ShiftInteger}
\end{figure}
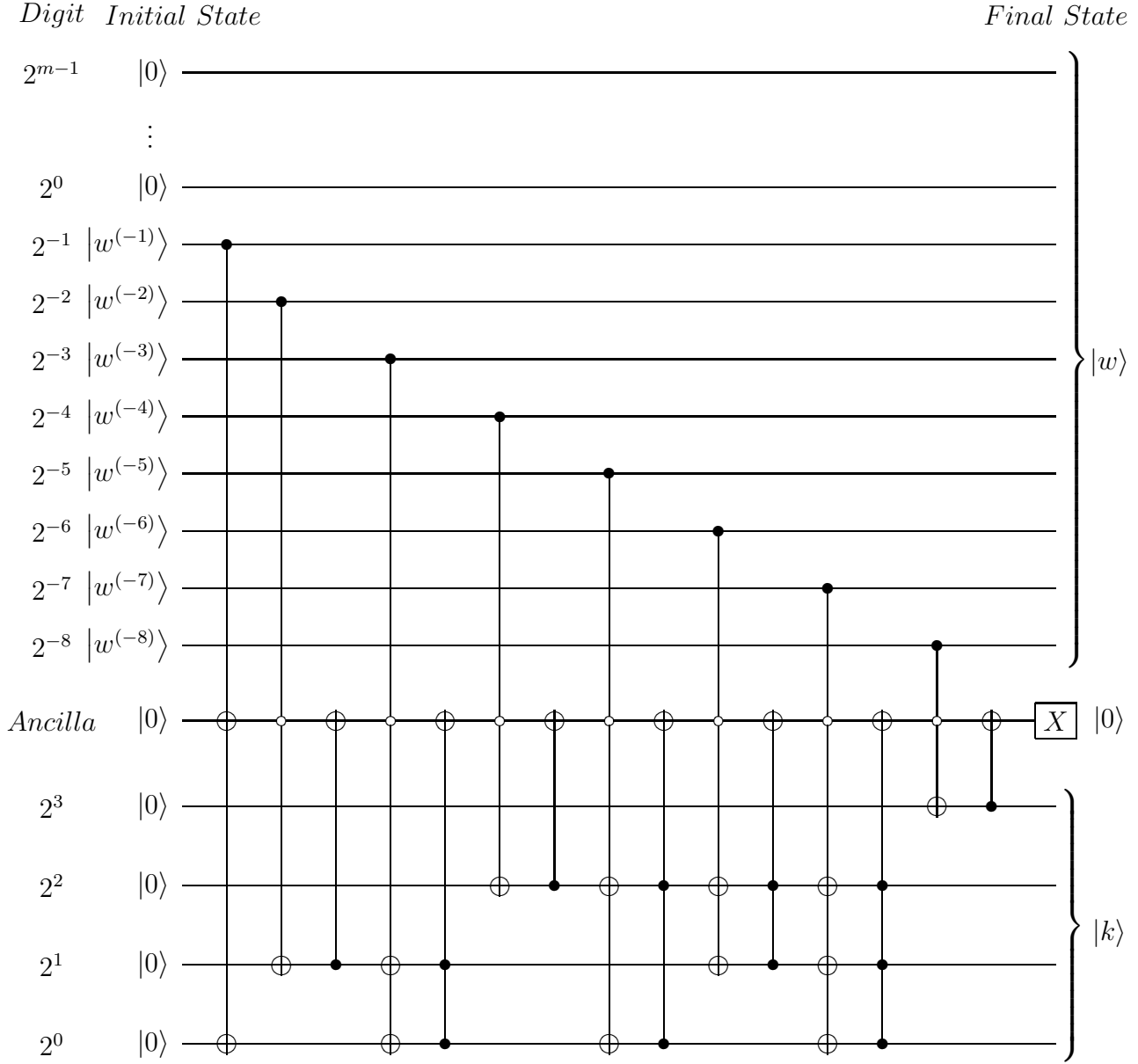

Now consider Algorithm \ref{alg:fracPower2} FractionalPower2, which requires two calls to Algorithm \ref{alg:fracPower} FractionalPower, one call to Algorithm \ref{alg:inv} INV, and 
a constant number of
calls to  a quantum circuit of the type shown in Fig. \ref{fig:basicModule}. 
Hence, using the previously derived bounds for the costs of each of these modules, the cost of Algorithm \ref{alg:fracPower2} FractionalPower2 in terms of both the number of quantum operations and the number of qubits is  a low-degree polynomial in $n,n_f$, and $\ell$, respectively.

\section{Numerical Results}  \label{sec:NumericalResults}

We present a number of tests comparing our algorithms to the respective ones implemented using floating point arithmetic in Matlab. To ensure that we compare against highly accurate values in Matlab we have used variable precision arithmetic (vpa) with 100 significant digits. In particular, we simulate the execution
of each of our algorithms and obtain the result. We obtain the corresponding result using a built-in function in Matlab. We compare the number of equal significant digits in the two results.
Our tests show  that using a moderate amount of resources, our algorithms compute values 
matching the corresponding ones obtained by commercial, widely used scientific computing software 
with 12 to 16 decimal digits. 
In our tests, the input $w$ was chosen randomly.

\subsection{Algorithm: SQRT}
In Table \ref{tab:numSQRT}, 
all calculations for SQRT are performed with $b = 64$ bits (satisfying $b > 2m$ for all inputs). 
The first column in the table gives the different values of $w$ in our tests, the second and third columns give the computed value using Matlab and our algorithm respectively, and the last column gives the number of identical significant digits between the output of our algorithm and Matlab. \\

\begin{table}[h]
\begin{center}
\begin{tabular}{| c | c | c | c |}
	\hline
	$w$ & Matlab: $w^{1/2}$ & Our Algorithm: $w^{1/2}$ & \# of Identical Digits \\
	\hline
	$0.0198$ & $0.140712472794703$ & $0.140712472794703$ & $16$ \\
	\hline
	$48$ & $6.928203230275509$ &  $6.928203230275507$ & $15$  \\
	\hline
	$91338$ & $302.2217728754829$ & $302.2217728754835$ & $14$  \\
	\hline
	$171234050$ & $13085.64289593752$ &  $13085.64289596872$ & $12$ \\
	\hline
\end{tabular}
\end{center}
\caption{Numerical results for SQRT.}
\label{tab:numSQRT}
\end{table}

In Fig. \ref{fig-sqrtplot} we give simulation results showing the error of our algorithm for different values of the parameter $b$. In particular, we show how the worst-case error of the algorithm depends on $b$, and then, using the results of Matlab as the correct values of $\sqrt w$, we show the actual error of our algorithm in relation to $b$ in a number of test cases.

\begin{figure}[H]
\centering
\includegraphics[scale=0.8]{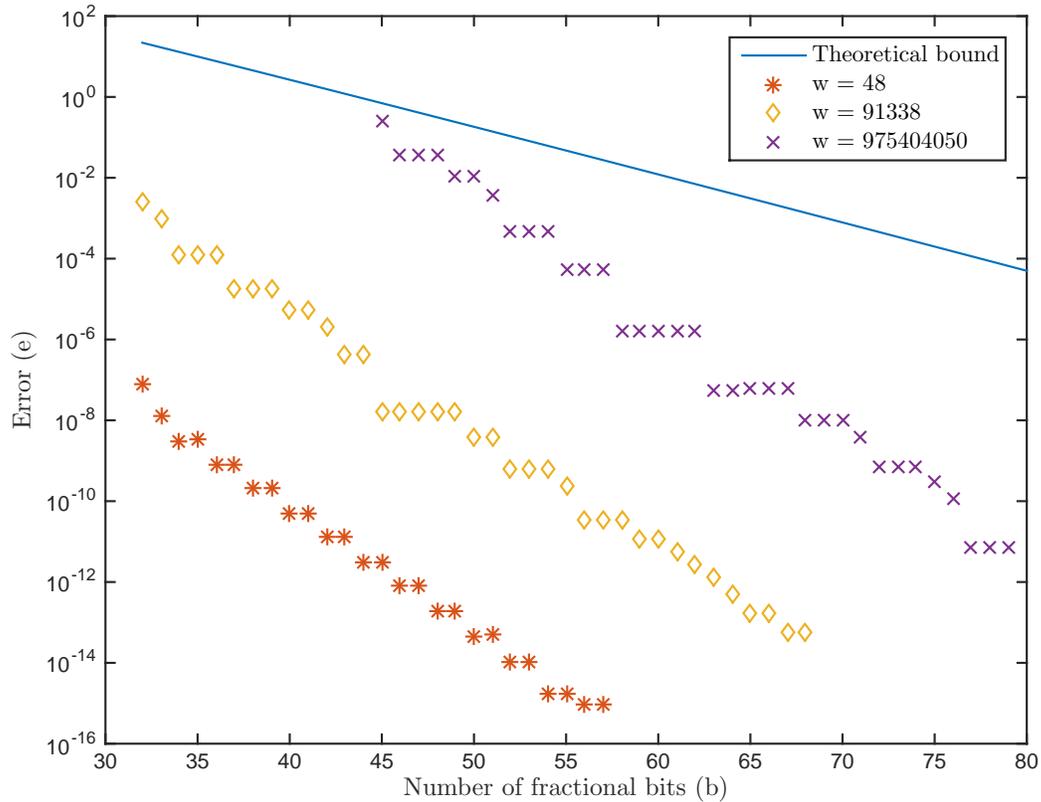}
\caption{Semilog plot showing the error ($e$) of Algorithm 1 SQRT versus the number of precision bits $b$. The solid blue line 
is a plot of the worst-case error of Theorem 1, for $n=2m=64$. 
The three data sets represent the absolute value of the difference between Matlab's floating point calculation of $\sqrt{w}$ 
and our algorithm's result for three different values of $w$.}
\label{fig-sqrtplot}
\end{figure}

\subsection{Algorithm: PowerOf2Roots}
In Table \ref{tab:numPow2Roots},
we show only the result for the $2^{k}$th root (the highest root), where $k$ was chosen randomly such that $5\le k \le 10$.  Again, all our calculations are performed with $b = 64$ bits. \\

\begin{table}[h]
\begin{center}
\begin{tabular}{| c | c | c | c | c |}
	\hline
	$w$ & $k$ & Matlab: $w^{1/2^{k}}$ & Our Algorithm: $w^{1/2^{k}}$ & \# of Identical Digits \\
	\hline
	$0.3175$ & $6$ & $0.982233508377946$ & $0.982233508377946$ & $16$ \\
	\hline
	$28$ & $10$ & $1.003259406317532$ & $1.003259406317532$ & $16$ \\
	\hline
	$15762$ & $5$ & $1.352618595919273$ & $1.352618595919196$ & $13$ \\
	\hline
	$800280469$ & $8$ & $1.083373703681284$ & $1.083373703681403$ & $13$ \\
	\hline
\end{tabular}
\end{center}
\caption{Numerical results for PowerOf2Roots.}
\label{tab:numPow2Roots}
\end{table}

\subsection{Algorithm: LN}
In Fig. \ref{fig-lnplot}  below we give simulation results showing the error of our algorithm for different values of the parameter $\ell$. In particular, we show how the worst-case error of the algorithm depends on $\ell$, and then, using the results of Matlab as the correct values of $\ln(w)$, we show the actual error of our algorithm in relation to $\ell$ in a number of test cases.  

\begin{figure}[H]
\centering
\includegraphics[scale=0.8]{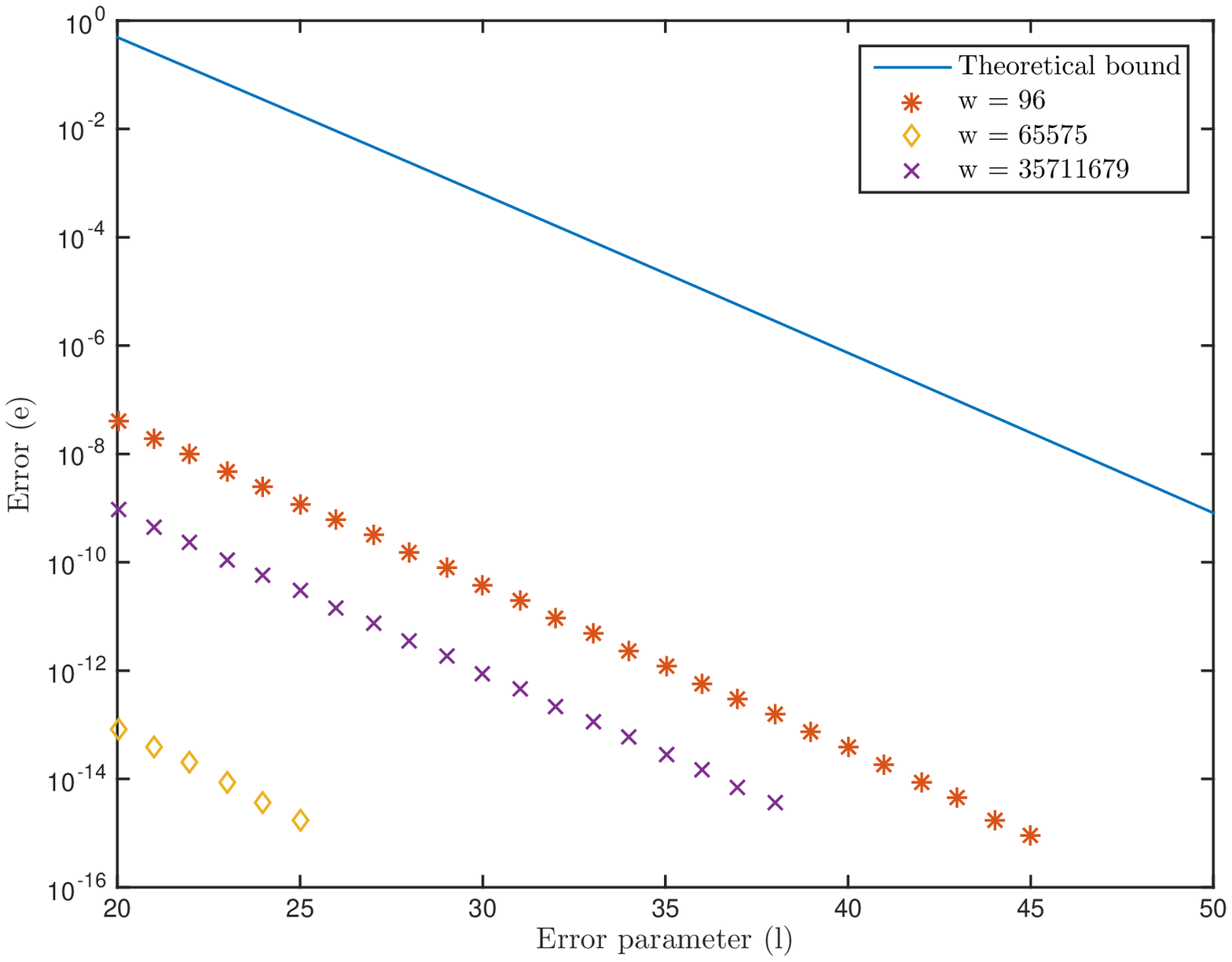}
\caption{Semilog plot showing the error of Algorithm 3 LN versus the user-defined parameter $\ell$ which controls the accuracy of the result. 
The solid blue line is a plot of the worst-case error of Theorem 3, for $n=2m=64$, $b=\max\{ 5\ell,25\}$. 
The three data sets represent the absolute value of the difference between Matlab's floating point calculation of $\ln(w)$ and our algorithm's result for three different values of $w$.}
\label{fig-lnplot}
\end{figure}

\begin{table}[h]
\begin{center}
\begin{tabular}{| c | r | r | c |}
	\hline
	$w$ & Matlab: $\ln(w) \quad$& Our Algorithm: $\ln(w)$ & \# of Identical Digits \\
	\hline
	$96$ & $4.564348191467836$ & $4.564348191467836$ & $16$ \\
	\hline
	$65575$ & $11.090949804735075$ & $11.090949804735075$ & $17$ \\
	\hline
	$35711679$ & $17.390988336107455$ & $17.390988336107455$ & $17$ \\
	\hline
\end{tabular}
\end{center}
\caption{Numerical results for LN.}
\label{tab:numLN}
\end{table}

In the algorithm that computes the logarithm, there is a user-determined parameter $\ell$ that controls the desired accuracy of the result. We have used $\ell=50$ in the tests shown in Table~\ref{tab:numLN}.

\subsection{Algorithm: FractionalPower}
The table below shows tests for the approximation of $w^f$ for randomly generated $f\in (0,1)$. This result of this algorithm also depends on  a user-determined parameter $\ell$ controlling the accuracy.
We have used $\ell=50$ in the tests shown in Table~\ref{tab:numFracPow}.

\begin{table}[h]
\begin{center}
\begin{tabular}{| c | c | r | r | c |}
	\hline
	$w$ & $f$ & Matlab: $w^f \qquad$ & Our Algorithm: $w^f$ & \# of Identical Digits \\
	\hline
	$0.7706$ & $0.1839$ & $0.953208384891996$ & $0.953208384891998$ & $15$ \\
	\hline
	$76$ & $0.7431$ & $24.982309269657478$ & $24.982309269657364$ & $14$ \\
	\hline
	$1826$ & $0.1091$ & $2.268975123215851$ & $2.268975123215838$ & $14$ \\
	\hline
	$631182688$ & $0.5136$ & $33094.79142555447$ & $33094.79142555433$ & $14$ \\
	\hline
\end{tabular}
\end{center}
\caption{Numerical results for FractionalPower.}
\label{tab:numFracPow}
\end{table}

\section{Implementation Remarks}   \label{sec:implementation}

In designing algorithms for scientific computing the most challenging task is to control the error 
propagation.  In our case, we derive the algorithms by combining elementary modules and provide worst-case error bounds.  
The resulting algorithms have cost that is polynomial in the problem parameters.  
In this section, we wish to discuss the implementation of the algorithms, and to give some estimates of their cost for certain choices of their parameters. 

Since we are interested in quantum algorithms, the algorithms must be reversible. 
To the extent that the modules are based on classical computations, their reversibility is not a major issue since \cite{bennett1989time, levine1990note, portugal2006reversible} show how to simulate them reversibly. There are time/space trade-offs in 
an implementation that are of theoretical and practical importance. 
These considerations, as well as other constraints such as (fault-tolerant) gate sets or locality restrictions \cite{Takahashi}, should be optimized by the compiler, but we are not concerned with them here. 
For a discussion of the trade-offs, see, for example, \cite{parent2015reversible} and the references therein. 
Nevertheless, for the algorithms in this paper we indicate 
a possible 
way to implement them, realizing that there are alternatives that one may opt for. 

The design of a library of algorithms for scientific computing must be modular, using components with a certain functionality and worst-case accuracy 
while hiding the implementation details.  The low-level implementation of the modules, i.e., actual 
derivation of the quantum circuits realizing them, should not be a concern for the scientist or the 
application programmer deriving the algorithms. In fact it should be possible to change the low-level 
implementation of the modules in a transparent way.  
The goal is to have a library of quantum circuit templates which can be instantiated 
as needed. 

We are at an early stage in the development of quantum computation, quantum programming languages and compilers. It is anticipated that implementation decisions will be made taking into account technological limitations concerning the different quantum computer architectures, but also it is expected that  things will change with time and some of the present constraints may no longer be as important.  

Although optimization of the individual quantum circuits realizing the algorithms for scientific 
computing described in this paper is desirable, to a certain degree, it is not a panacea.  Quantum 
algorithms may use one or more of our modules in different ways while performing their tasks. 
Therefore, global optimization of the resulting quantum circuit and resource management is a much 
more important and will have to be performed by the compiler and optimizer. One additional 
possibility would be to have templates for multiple versions of quantum circuit implementations of each of our quantum algorithms which can be selected according to optimization goals set during 
compilation. 
Quantum programming languages and compilers is an active area of research with many open questions. 
We do not pursue this direction since it is outside the scope of this paper.

We now present some implementation ideas for our algorithms. We start our discussion with the elementary module of Fig. \ref{fig:basicModule} because all algorithms use it. 
%
The quantum circuit realizing this module needs to perform multiplication, and addition.  
Let $a,b$ be fixed precision numbers held in $n$ qubit registers. One can perform an addition \lq\lq in-place\rq\rq, i.e., $\ket{a}\ket{b} \to \ket{a}\ket{a+b}$,  with the understanding that the second register size has to be $n+1$ qubits long to hold the result exactly. A multiplication can be computed \lq\lq out-of-place\rq\rq, i.e., $\ket{a}\ket{b}\ket{0} \to \ket{a}\ket{b}\ket{a+b}$. This requires an additional $2n$ qubits to hold the result. The same could have been done for the addition but performing it in-place saves some space.   Similar considerations apply in a straightforward way when the inputs are held in registers of different sizes. 

In the literature there are numerous ways that 
basic arithmetic operations can be performed reversibly, and there are trade-offs between the resulting quantum circuits in terms of the number of ancilla bits/qubits used and the resulting circuit size and depth; for example, see 
\cite{cuccaro2004new, draper2000addition, DKRS06, kepleyquantum, portugal2006reversible, rieffel2011quantum, saeedi2013synthesis, TK05, Takahashi, van2005fast, takahashi2008fast}.
 Table 1 in \cite{Takahashi} summarizes some of the trade-offs.  
 There we see that the circuit in \cite{draper2000addition} does not use Toffoli gates or ancilla qubits for addition, but has size $O(n^2)$, while  the circuit in \cite{cuccaro2004new} uses one ancilla qubit, $O(n)$ Toffoli gates, and has size $O(n)$. We believe that low-level implementation details have to be decided taking into account the target architecture because the unit cost of the different resources is not equal.  As we mentioned, one may wish to have different circuits for the same arithmetic operation, and at compile time to let the optimizer decide which one is preferable.

In summary, the circuit implementing the elementary module of Fig. \ref{fig:basicModule} contains sub-circuits for addition and multiplication, plus ancilla registers to facilitate these operations, some of which are cleared at the end of the module application and can be reused, while others save intermediate results and may be cleared at a later time. 

Our algorithms, in addition to the applications of the elementary module of Fig. \ref{fig:basicModule}, use a constant number of extra arithmetic operations. 
Thus the cost of the algorithms can be expressed in a succinct and fairly accurate way by counting the number of arithmetic operations required, which is the approach taken in classical algorithms for scientific computing.
The number of arithmetic operations can provide precise resource estimates once particular choices for the quantum circuits implementing these operations have been made.
We remark that the costs of the quantum circuits of figures 
\ref{fig-InitState1}, \ref{fig-InitState2}, \ref{fig-RShift}, \ref{fig-compute_p}, and \ref{fig-ShiftInteger}
are subsumed by the cost of the arithmetic operations. 


Our algorithms use Newton Iteration. 
Each iterative step is realized by cascading a small number of modules like the one of Fig. \ref{fig:basicModule}, whose implementation we discussed above. 
In addition, at each iterative step, we keep the inputs and the truncated output which is the next approximation. 
This information alone makes the step reversible. 
Note that keeping the intermediate results introduces a relatively insignificant overhead since Newton iteration converges quadratically, and thus the number of iterations is proportional to the logarithm of the bits of accuracy, i.e $\log_2 b$. For example, for $b=64$, keeping the intermediate results of Newton iteration requires storing $6$ fixed-precision numbers. The
number of intermediate results kept by algorithms that apply Newton iteration multiple times, such as Algorithm \ref{alg:sqrt} SQRT that applies Newton iteration twice, 
and the algorithms that require the computation of $2^k$-roots, 
is derived accordingly. 

For the convenience of the reader, we provide some illustrative cost and qubit estimates. 
However, these numbers must be interpreted keeping in mind the larger picture. Our functions are expected to be used as subroutines. 
Note, the cost of a quantum algorithm that computes one or more of our functions a number of times in sequence is quite different to that of an algorithm that performs the same computations concurrently. 

The tables below give estimates of the cost of our algorithms 
for specific choices of the problem parameters. 
As explained, we assume the intermediate approximations of each step of Newton iteration are kept in memory until the end of the computation. It is thus possible to uncompute all intermediate approximations if necessary with the same number of qubits and at most a doubling the number of arithmetic operations. 

\begin{table}[h]   
\begin{minipage}{\textwidth}
\begin{center}
\begin{tabular}{| c | c | c | c |}
	\hline
	Algorithm & \# Multiplications & \# Additions & \# Qubits \\
	\hline
	INV                                  & \ \ $20$ & \ $10$ &  \ $370$ \\  
	\hline
	SQRT                               & \ \ $60$ & \ $20$ &  \ $790$ \\  
	\hline
	PowerOf2Roots($k=5$)   & \ $600$ & $200$  & \ $910$ \\  
	\hline
	PowerOf2Roots($k=10$) & $1200$ & $400$  &  $1110$ \\  
	\hline
\end{tabular}
\end{center}
\caption{Number of arithmetic operations and qubit estimates rounded to the nearest $10$ for $n=2m=b=32$.}
  \label{table:costs}
\end{minipage}
\end{table}

Let us discuss the entries of Table \ref{table:costs}, where we have selected $n=2m=b=32$.  
Consider Algorithm \ref{alg:inv} INV. There are $ \log b =5$ iterative steps. Each step has two multiplications and one addition (ignoring the multiplication by $2$) as seen in (\ref{eq:NI.INV}). At each iterative step, clearing the ancilla registers, with the exception of the register holding the truncated next approximation, doubles the number of arithmetic operations. The number of qubits follows from the fact that we have an initial approximation and $5$ successive approximations, each held in a $b=32$ qubit register, plus we have $n=32$ qubits holding the input $w$, plus we use about $150$ ancilla qubits, which are sufficient to store the intermediate results, and which are reused at each iterative step. 

Algorithm \ref{alg:sqrt} SQRT follows similarly by observing that it applies five steps of Newton iteration twice. The second iteration function (\ref{eq:NI.SQRT}) requires two more multiplications than iteration function (\ref{eq:NI.INV}) and the same number of additions. Just like before, at each iterative step we keep the truncated next approximation and uncompute the other ancilla qubits used within the step. To estimate the number of qubits, observe that we keep twice as many intermediate truncated approximations as before, 
and just over $300$ qubits are sufficient to hold the other intermediate results exactly in each iterative step. 

Algorithm \ref{alg:2krt} PowerOf2Roots$(k)$ computes $k$ square roots in turn. 
After each square root call we uncompute the ancilla registers holding intermediate approximations to economize on qubits. 
Thus, the number of arithmetic operations is $2k$ times that of Algorithm~\ref{alg:sqrt}~SQRT, and the number of qubits 
is given by the number of qubits needed for SQRT, plus sufficiently many qubits to hold the $k-1$ additional outputs (recall that this algorithm outputs approximations for for all $w^{1/2^{k}}$, $=1,\dots, k$).


For Algorithm~\ref{alg:ln}~LN, we compute with high accuracy 
the $1/2^\ell$-power of the 
shifted input in the range $[1,2)$. Then we use the first two terms of the Taylor series to save on computational cost (i.e. no divisions or powers beyond squaring are required). The dominant cost factor is the PowerOf2Roots subroutine. 
We give cost estimates for some values of $\ell$. Note that $\ell$ is an input that determines the desired accuracy of the result and number of significant bits $b$ used in the calculation. 

\begin{table}[h]  
\begin{minipage}{\textwidth}
\begin{center}
\begin{tabular}{| c | c | c | c | c |}
	\hline
	$\ell$ & $b$ &  \# Multiplications & \# Additions & \# Qubits \\
	\hline
	10 & \ 50 & $1440$ & \ $480$ & $1780$ \\
	\hline
	15 & \ 75 & $2520$ & \ $840$ & $3160$\\
	\hline
	20 & 100 & $3360$ & $1120$ & $4690$\\
	\hline
\end{tabular}
\end{center}
\caption{Number of arithmetic operations and qubit estimates rounded to the nearest $10$ for Algorithm \ref{alg:ln} LN for $n=2m=32$ and various values of the error parameter $\ell$. 
Note that $b=5\ell$ as in the algorithm.}
 \label{table:costLN}
\end{minipage}
\end{table}

Similarly, the cost and qubits of Algorithm \ref{alg:fracPower} FractionalPower follows from 
that of Algorithm \ref{alg:2krt} PowerOf2Roots. The multiplication steps at the end of the algorithm require extra qubits for the intermediate products. Algorithm \ref{alg:fracPower2} FractionalPower2 calls Algorithm \ref{alg:fracPower} FractionalPower twice. To save qubits, one can assume the intermediate results of Algorithm \ref{alg:fracPower} FractionalPower are uncomputed after each application.

\begin{table}[h]   
\begin{minipage}{\textwidth}
\begin{center}
\begin{tabular}{| c | c | c | c |}
	\hline
	Algorithm & \# Multiplications & \# Additions & \# Qubits \\
	\hline
	FractionalPower & \ $580$ & $190$ & $4020$ \\
	\hline
	FractionalPower2 & $2330$ & $780$ &  $4330$ \\
	\hline
\end{tabular}
\end{center}
\caption{Number of arithmetic operations and qubit estimates rounded to the nearest $10$ for $n=2m=16$, $n_f = 3$, and $\ell = 10$.}
  \label{table:costs2}
\end{minipage}
\end{table}

The respective numbers of arithmetic operations and qubits estimated for Algorithm \ref{alg:fracPower} FractionalPower and Algorithm \ref{alg:fracPower2} FractionalPower2 are larger 
than those of the previous algorithms.
There are two reasons for this.  
The first reason is that in these algorithms the error accumulates in two ways. 
One is the error in computing the $2^{i}$-roots, $i=1,\dots,n_f$, and the other is the error in multiplying these results. Thus we must use a large number of accuracy bits $b$ to control the overall error. 
The second reason is that the worst-case error estimates of Theorems \ref{thm4} and \ref{thm5}, from which the costs follow,  are conservative, 
since we have opted for simplicity in the error bounds presented.
 Cost improvements can be obtained by tightening the error bounds at the expense of more complicated error bound formulas. 
This can be done using the details already found in the theorem proofs. 


\section{Summary}  \label{sec:Discussion}

Recent results \cite{Poisson,Harrow,TaSma} suggest that quantum computers may have a substantial advantage over classical computers for solving numerical  linear algebra problems and, more  generally, problems of computational science and engineering. Scientific computing applications require the evaluation of elementary functions like those found in mathematics libraries of programming languages, where the calculations are performed using floating point arithmetic.
Although in principle quantum computers can always directly simulate any classical algorithm, generally there is no guarantee that such simulations remain practical. 
Hence, we need to develop efficient quantum algorithms and circuits implementing such functions and to provide performance and cost guarantees. 
It is also important to eventually develop a standard for the implementation of quantum algorithms for scientific computing similar to that of IEEE 754-2008 for floating point arithmetic.

In the design of algorithms for scientific computing,  
the challenging task is 
to control the error propagation and to do so efficiently. 
 In this sense, it is important to derive reusable modules with well-controlled error bounds. 
In this paper, we have given efficient, reversible, and scalable quantum algorithms for computing power of two roots, the logarithm, and the fractional power functions common to scientific computing. In particular, we have shown our algorithms to have both controllable error and reasonable (i.e. polynomially-bounded) cost in terms of the number of qubits and quantum operations required. The design is modular, combining a number of elementary quantum circuits to implement the functions. Each new algorithm builds on the previously developed ones as much as possible. 

This paper extends the results of \cite{Poisson}, where we exhibit quantum circuits computing the reciprocal, the sine, the cosine and their inverses. As explained, these circuits give fundamental building blocks critical to actual implementations of quantum algorithms involving numerical computations, such as those of \cite{Poisson,Harrow}. Our goal is to extend our work and ultimately develop a comprehensive library of quantum circuits for scientific computing.

\section*{Acknowledgements}

This work was supported in part by NSF/DMS.


\section*{Appendix}  \label{sec:Appendix}
In this section, we derive theorems showing the worst-case error of the algorithms in Section 2. \\

\noindent The following Corollary follows from Theorem B.1 of \cite{Poisson}.

\begin{cor} 
\label{cor0}
For $w>1$, represented by $n$ bits of which the first $m$ correspond to its integer part, and for $b\geq n$, 
 Algorithm \ref{alg:inv} returns a value $\hat{x}_{s}$ approximating $\frac{1}{w}$ with error
\begin{equation}
|\hat{x} - \frac{1}{w}| \leq \frac{2+ \log_2 b}{2^b}
\end{equation}
This is accomplished by performing Newton's iteration and truncating the results of each iterative step to $b$ bits after the decimal point.
\end{cor}

\begin{proof}
From \cite{Poisson}, 
the Newton iteration of Algorithm \ref{alg:inv} performed with $b$ bits of accuracy and $s$ iterations, produces an approximation $\hat{x}$ of $\frac{1}{w}$ such that
$$ |\hat{x} - \frac{1}{w}| \leq \left( \frac{1}{2} \right)^{2^s} + s2^{-b}.$$
For $s=\lceil \log_2 b \rceil$, we have  
$$ |\hat{x} - \frac{1}{w}| \leq \frac{1}{2^b} + \lceil \log_2 b \rceil \frac{1}{2^b} \leq  \frac{1}{2^b}  (2  +\log_2 b) .$$
\end{proof}

\begin{theorem} 
\label{thm1}
For $w>1$, represented by $n$ bits of which the first $m$ correspond to its integer part, and for $b\geq \max\{2m,4\}$, Algorithm \ref{alg:sqrt} returns a value $\hat{y}_{s}$ approximating $\sqrt{w}$ with error
\begin{equation}
|\hat{y}_{s}-\sqrt{w}|  
\leq \left( \frac{3}{4} \right)^{b-2m}  \left( 2+ b + \log_2 b  \right).
\end{equation}
This is accomplished by performing Newton's iteration and truncating the results of each iterative step to $b$ bits after the decimal point before passing it on the next iterative step.
\end{theorem}

\begin{proof}
The overall procedure consists of two stages of Newton's iteration, as illustrated in Fig. \ref{fig-SQRToverall} above. We analyze each stage in turn.

Observe that the iteration $x_i = g_1(x_{i-1}):=-w\hat{x}_{i-1}^2 + 2\hat{x}_{i-1}$, $i=1,2,\dots,s_1$, $s_1=\lceil \log_2 b \rceil$, corresponds to Newton's iteration applied to the function $f_1\left(x\right) := \frac{1}{x}-w$  for approximating $\frac{1}{w}$, with initial guess $\hat{x}_0 = 2^{-p}$ where $p\in \nat$ and $2^p > w \geq 2^{p-1}$. 
It has been analyzed in detail in Theorem B.1 of \cite{Poisson}. Here, we briefly review some of the results. 
An efficient circuit for generating the initial state $2^{-p}$ is shown in Fig. \ref{fig-InitState1} above, similar to that in \cite{Poisson}.

Note that the approximations $x_i  < \frac{1}{w}$, i.e. the iteration underestimates $\frac{1}{w}$. Indeed, $g_1\left(x\right) - \frac{1}{w} = \frac{1}{w}\left(2xw - w^2x^2-1\right) = -\frac{1}{w}\left(wx-1\right)^2 < 0$.  Taking into account the truncation error, we have 
\begin{equation}
\label{eq:InverseError}
|\xh - \frac{1}{w} | \leq \left(we_0\right)^{2^{s_1}}\frac{1}{w} + 2^{-b}s_1  \leq  \left(\frac{1}{2}\right)^{2^{s_1}} + 2^{-b}s_1 =:E .
\end{equation}
This follows from the facts $w>1$ and $we_0 \leq \frac12$, where $e_0= |2^{-p} - \frac{1}{w}|$, and is shown in \cite{Poisson}.
 The first term in the upper bound corresponds to the error of Newton's iteration, 
while the second term is the truncation error.  
Now we turn to the second stage. 
Iteration $y_j = g_2(y_{j-1}):= \frac12(3y_{j-1} - \hat{x}_{s_1} y^3_{j-1} )$, $j=1,2,\dots,s_2$, $s_2= \lceil \log_2 b \rceil$, is obtained by using Newton's method to approximate the zero of the function $f_2\left(y\right)= \frac{1}{y^2}-\x$, with initial guess $\hat{y}_0 = 2^{\lfloor \frac{q-1}{2} \rfloor}$ where $q\in \nat$ and $2^{1-q} > \hat{x}_{s_1} \geq 2^{-q}$. An efficient circuit for generating the initial state $2^{\lfloor \frac{q-1}{2} \rfloor}$ is shown in Fig. \ref{fig-InitState2}. We have
\begin{eqnarray*}
g_2\left(y\right) - \frac{1}{\sqrt{\x}} &=& y - \frac{1}{\sqrt{\x}} + \frac{1}{2}\left(y-y^3\x\right)\\
&=&  y - \frac{1}{\sqrt{\x}} + \frac{\x y}{2}\left(\frac{1}{\x}-y^2\right)\\
&=&  \left(y- \frac{1}{\sqrt{\x}}\right)   \left(  1 - \frac12 \x y\left(y+ \frac{1}{\sqrt{\x}}\right)   \right)      \\     
&=& - \frac12   \left(y- \frac{1}{\sqrt{\x}}\right)  \left( y^2\x +y\sqrt{\x} -2\right)\\
&=&  - \frac12   \left(y- \frac{1}{\sqrt{\x}}\right)  \left(  \left(y\sqrt{\x} -1 \right)^2  + 3y\sqrt{\x} -3  \right)\\
&=&  - \frac12   \left(y- \frac{1}{\sqrt{\x}}\right)   \left(y- \frac{1}{\sqrt{\x}}\right) \sqrt{\x}   \left(y\sqrt{\x}  -1 + \frac{3y\sqrt{\x} -3}{y\sqrt{\x}  -1}\right) \\
&=&    - \frac12   \left(y- \frac{1}{\sqrt{\x}}\right)^2  \sqrt{\x} \left(y\sqrt{\x} +2\right) \\
\end{eqnarray*}
The last quantity is non-positive assuming $y\geq 0$.
The iteration function $g_2$ is non-decreasing in the interval $[0,\frac{1}{\sqrt{\x}}]$ and $g_2\left(0\right)=0$. 
Since $y_0 = 2^{-\lfloor \frac{q-1}{2} \rfloor} $, we get $y_1 \geq 0$, and inductively we see that all iterations produce positive numbers that are approximations underestimating $\frac{1}{\sqrt{\x}}$, i.e. $y_i\leq \frac{1}{\sqrt{\x}}$ for $i=0,1,2,\dots$.
Then
\begin{equation}
e_{i+1} := |y_{i+1} - \frac{1}{\sqrt{\x}} | = e_i^2 \frac12 \sqrt{\x} | y_i \sqrt{\x}  +2  | \leq \frac32 \sqrt{\x} e_i^2,
\end{equation}
since $| y_i \sqrt{\x}  +2  | \leq 3$ ($y_i$ underestimates $\xh^{-\frac12})$,
where $e_0 = \frac{1}{\sqrt{\x} } - 2^{-\lfloor \frac{q-1}{2} \rfloor}  \leq 2^{-\lfloor \frac{q-1}{2} \rfloor}$.

Let $A=\frac32 \sqrt{\x}$. We unfold the recurrence to obtain
$e_i \leq \frac{1}{A} {\left(Ae_0\right)}^{2^i} $, $i=1,2,\dots$. We have $\sqrt{\x} 2^{\lfloor \frac{q-1}{2} \rfloor} \geq 2^{-q/2} 2^{\lfloor \frac{q-1}{2} \rfloor} $. For $q$ odd, this quantity is bounded from below by $\frac{1}{\sqrt{2}}$, and for $q$ even this is bounded by $\frac12$. 
Thus $\sqrt{\x}e_0 = |\sqrt{\x} 2^{\lfloor \frac{q-1}{2} \rfloor}  - 1 | = 1 - \sqrt{\x} 2^{\lfloor \frac{q-1}{2} \rfloor}   \leq \frac12$ because we have selected the initial approximation $y_0$ to underestimate $\frac{1}{\sqrt{\x}}$. From this, we obtain $e_i \leq \frac{1}{A}\left(\frac{3}{4}\right)^{2^i}$.
Using equation (\ref{eq:InverseError}), we have that $\xh \geq \frac{1}{w}-E$. 
Without loss of generality, $wE \leq \frac12$. This will become apparent in a moment once we select the error parameters. Thus, $\frac1A \leq \frac23 \sqrt{2w}$. Therefore,
\begin{equation}
e_i \leq \frac{2}{3} \sqrt{2w} \left( \frac{3}{4} \right)^{2^i}, \;\; i=1,\dots,s_2.
\end{equation}

We now turn to the roundoff error analysis of the second stage of the algorithm. The iterations of the second stage of the algorithm would produce a sequence of approximations $y_1, \dots , y_{s_2}$ if we did not have truncation error. Since we truncate the result of each iteration to $b$ bits of accuracy before performing the next iteration, we have a sequence of approximations $\hat{y}_1, \dots , \hat{y}_{s_2}$, with
\begin{eqnarray*}
\hat{y}_{0} &=& {y}_{0}\\
\hat{y}_{1} &=& g\left(\hat{y}_{1}\right) + \xi_1\\
\vdots \\
\hat{y}_{i} &=& g\left(\hat{y}_{i}\right) + \xi_i\\
\vdots
\end{eqnarray*}
where $|\xi_i | \leq 2^{-b}$, $i\geq 1$.
Using the fact that $\sup_{x\geq 0} |g_2'\left(y\right)| \leq \frac32$, we obtain
\begin{eqnarray*}
|\hat{y}_{s_2}-y_{s_2}| &\leq& |g_2\left(\hat{y}_{s_2 -1}\right) - g_2\left(y_{s_2 -1}\right)| + |\xi_{s_2}|\\
 &\leq& \frac32| \hat{y}_{s_2 -1} - y_{s_2 -1}| + |\xi_{s_2}|\\
 &\vdots& \\
 &\leq& \sum_{j=1}^{s_2} \left(\frac32\right)^{s_2 -j} | \xi_j |
  \leq 2^{-b} \frac{ \left(\frac32\right)^{s_2} -1 }{\frac32 -1}\\
    &\leq& 2^{1-b} \left(\frac32\right)^{s_2}. 
\end{eqnarray*}
Therefore, the total error of the second stage of the algorithm is 
\begin{equation}
|\hat{y}_{s_2} - \frac{1}{\sqrt{\x}}| \leq     \frac{2}{3} \sqrt{2w} \left( \frac{3}{4} \right)^{2^{s_2}}  +  2^{1-b}\left(\frac32\right)^{s_2}.
\end{equation}
For $b \geq \max\{2m,4 \}$ and  $s_2 = \lceil \log_2 b \rceil$, and recall that $w\leq 2^m$. Then we have
\begin{eqnarray}
\label{eq:Stage2Error}
|\hat{y}_{s_2} - \frac{1}{\sqrt{\x}}| &\leq&   \frac{2}{3} \sqrt{2} \; 2^{\frac{m}2} \left( \frac{3}{4} \right)^{b}  +  2^{1-b}\left(\frac32\right)^{\log_2 b +1}   \nonumber \\
&\leq&  \frac{2}{3} \sqrt{2} \; \left(\sqrt{2}\right)^{m} \left( \left( \frac{3}{4} \right)^2\right)^{b/2}  +  2^{1-b}\; 2^{\log_2 b +1} \nonumber \\
&\leq&  \frac{2}{3} \sqrt{2} \; \left(   \frac{\sqrt{2}\;9}{16} \right)^m \left( \frac{9}{16} \right)^{b/2-m}  +  2^{2-b} \; b \nonumber \\
&\leq&  \left( \frac{3}{4} \right)^{b-2m}  +   2^{2-b} \; b,
\end{eqnarray}
Let us now consider the total error of our algorithm,
\begin{equation} \label{eq:TotalError}
|\hat{y}_{s_2}-\sqrt{w}| \leq | \hat{y}_{s_2}- \frac{1}{\sqrt{\x}}| + |  \frac{1}{\sqrt{\x}}-\sqrt{w}|
\end{equation}
We use equation (\ref{eq:Stage2Error}) above to bound the first term. For the second term we have
\begin{equation}  \label{eq:secondTerm}
 |  \frac{1}{\sqrt{\x}}-\sqrt{w}| \leq \frac12|\frac{1}{\x}-w| 
=   \frac12 \frac{w}{\x} |\x - \frac{1}{w}|
 \leq \frac12 \frac{w}{\x} E,
\end{equation}
where the first inequality follows from the mean value theorem ($|\sqrt{a} - \sqrt{b}| \leq \frac12 |a-b|$ for $a,b \geq 1$).
Since $\x \geq \frac{1}{w}-E$,
\begin{equation}
 \x^{-1} \leq \left(1 - \frac{1}{w} - E\right)^{-1} \leq 2w,
\end{equation}
for $wE \leq \frac12$. Then 
(\ref{eq:secondTerm}) becomes 
\begin{eqnarray}
 |  \frac{1}{\sqrt{\x}}-\sqrt{w}| &\leq& w^2 E \leq  w^2 \left( \left(\frac{1}{2}\right)^{2^{s_1}} + 2^{-b}s_1 \right) \nonumber\\
  &\leq&  2^{2m} \left( \left(\frac{1}{2}\right)^{2^{s_1}} + 2^{-b}s_1 \right) \nonumber\\
    &\leq&    2^{2m} \left(  2^{-b} + 2^{-b}  \lceil \log_2 b \rceil \right) \nonumber\\
        &\leq&    2^{2m-b}\left( 1+  \lceil \log_2 b \rceil \right),
\end{eqnarray}
where we set $s_1 = \lceil \log_2 b \rceil$. 
Combining this with equation (\ref{eq:TotalError}), 
\begin{eqnarray}
|\hat{y}_{s_2}-\sqrt{w}|     &\leq &  \left( \frac{3}{4} \right)^{b-2m}  +   2^{2-b} \; b+    2^{2m-b}  \left( 1+  \lceil \log_2 b \rceil \right) \nonumber\\
&\leq&   \left( \frac{3}{4} \right)^{b-2m}  \left( 2+ b + \log_2 b  \right),
\end{eqnarray}
and the error bound in the statement of the theorem follows for $s=s_1=s_2$.

Finally, for completeness we show that for $b\geq \max\{2m,4\}$, $wE\leq \frac12$. Indeed, 
\begin{eqnarray*}
wE &\leq& 2^m \left(  2^{-b} + 2^{-b}  \lceil \log_2 b \rceil \right)\\
 &\leq& 2^{-b+m}   \left( 2+  \log_2 b \right)\\
 &\leq& \left( \frac12\right)^{b/2-m}   \frac{2+ \log_2 b}{2^{b/2}}.
 \end{eqnarray*}
The first factor above is at most $\frac12$ since $b\geq \max\{2m,4\}$, while the second is at most $1$ and this completes the proof.
\end{proof}

\begin{theorem} 
\label{thm2}
For $w>1$, represented by $n$ bits of which the first $m$ correspond to its integer part, Algorithm \ref{alg:2krt} computes approximations $\hat z_1$,$\hat z_2$,\dots,$\hat z_k$ such that
\begin{equation} \label{eq:thm2}
|\hat{z}_{i}-w^{\frac{1}{2^i}}| \leq 2  \left( \frac{3}{4} \right)^{b-2m}  \left( 2+ b +  \log_2 b  \right), \; i=1,2,\dots, k, \rm{\ for\ any\ } k\in \nat.
\end{equation}
This is accomplished by repetitively calling Algorithm \ref{alg:sqrt}.  Each $z_i$ has $b\ge \max\{ 2m, 4\}$ bits after its decimal point, and by convention its integer part is $m$ bits long.
\end{theorem}

\begin{proof}
We have 
\begin{eqnarray*}
\hat z_1 &=& \sqrt{w} + \xi_1  \\
\hat z_2&=& \sqrt{z_1} + \xi_2  \\
&\vdots& \\
\hat z_k &=& \sqrt{z_k} + \xi_k. 
\end{eqnarray*}
Since $\hat z_i$ is obtained using Algorithm \ref{alg:sqrt} with input $\hat z_{i-1}$, we use Theorem \ref{thm1} to obtain $|\xi_i| \leq \left( \frac{3}{4} \right)^{b-2m}  \left( 2+ b + \log_2 b  \right)$, $i=1,2,\dots,k$. 
We have
\begin{eqnarray*}
|\hat z_k - w^{\frac{1}{2^k}}| &\leq& |\sqrt{\hat z_{k-1}} - w^{\frac{1}{2^k}}| + |\xi_k|\\
&\leq& \frac12 |\hat z_{k-1} - w^{\frac{1}{2^{k-1}}}| + |\xi_k|\\
&\vdots& \\
&\leq& \sum_{j=0}^{k-1}  \frac{1}{2^j} |\xi_{k-j}| \\
&\leq& 2  \left( \frac{3}{4} \right)^{b-2m}  \left( 2+ b + \log_2 b  \right), 
\end{eqnarray*}
where the second inequality above follows again from $|\sqrt{a} - \sqrt{b}| \leq \frac12 |a-b|$ for $a,b \geq 1$.
\end{proof}

\begin{theorem} 
\label{thm3}
For $w>1$, represented by $n$ bits of which the first $m$ bits correspond to its integer part, Algorithm \ref{alg:ln} computes an approximation ${\hat{z}:= \hat{z}_p + (p - 1) r}$ of $\ln w$, where $2^p > w \geq 2^{p-1}$, $p\in \nat$, and $|r-\ln 2|\le 2^{-b}$, such that 
$$ |  \hat{z}  - \ln w | \leq   \left(\frac{3}{4}\right)^{5\ell/2} \left( m+ \frac{32}{9} + 2\left(\frac{32}{9} + \frac{n}{\ln 2} \right)^3 \right),$$
where $\ell\ge \lceil \log_2 8 n\rceil$ is a parameter specified in the algorithm that is used to determine the number $b=\max\{5\ell, 25\}$ of bits after the decimal point in which arithmetic is be performed, and from that the error. 
\end{theorem}

\begin{proof}
An overall illustration of the algorithm is given in Fig. \ref{fig-LogOverall}.
 
Our algorithm utilizes the identity $\ln w = \ln2 \log_2 w$, as well as other common properties of logarithms. For completeness, an example circuit for computing $p$ is given in Fig. \ref{fig-ShiftInteger} above. We proceed in detail. 

If the clause of the second \textit{if} statement evaluates to true, 
in the case $w=2^{p-1}$, 
then the algorithm sets $z_p = 0$ and 
returns $(p-1)r$. This quantity approximates $\ln 2^{p-1}$ with error bounded from above by $(m-1)2^{-b}$, since $p\le m$ in the algorithm.  
We deal with the case 
$w$ not a power of $2$, for which  $z_p\ne 0$.  
Using the same notation as the algorithm, we have
\begin{eqnarray}
\label{eq:lnErrTot}
 | z_p + (p -1)r - \ln(w)| &\le& |z_p + \ln 2^{p-1} - \ln(w)| + |(p-1)r - \ln 2^{p-1}| \nonumber\\
 &=& |z_p - \ln 2^{-(p-1)}w| +(m-1) 2^{-b}   \nonumber \\ 
&\leq&  |z_p - 2^\ell \ln w_p^{\frac{1}{2^\ell}}| + (m-1) 2^{-b}\nonumber \\
 &\leq&  |z_p -   2^{\ell} \ln \hat{t}_p | + | 2^\ell \ln \hat{t}_p- 2^\ell \ln w_p \frac{1}{2^\ell}| + (m-1)2^{-b}\nonumber\\
 &\leq&  | 2^{\ell} \hat{y}_p  - 2^{\ell} y_p | +  | 2^{\ell} y_p - 2^\ell \ln \hat{t}_p| \\
&\qquad& \qquad+   | 2^\ell \ln \hat{t}_p- 2^\ell \ln w_p^{\frac{1}{2^\ell}}| +(m-1)2^{-b}\nonumber,
\end{eqnarray}
where $w_p = 2^{1-p}w$,   
$z_p = 2^\ell y_p$, 
$\hat{z}_p = 2^\ell \hat{y}_p$ and 
by $y_p$ we denote the value $(\hat{t}_p-1) - \frac12 (\hat{t}_p -1)^2$ before it is truncated to obtain $\hat{y}_p$. The first term is due to truncation error and is bounded from above by $2^{\ell-b}$. The last term is bounded from above by $2^\ell |\hat{t}_p - w_p^{\frac{1}{2^\ell}}|;$ this is obtained using the mean value theorem for $\ln$ with argument greater than or equal to one. The second term is bounded from above by $2^{\ell}|2(\hat{t}_p -1)^3|$. Indeed, recall that $\hat{t}_p - 1 = \delta \in (0,1)$ according to line 12 of the algorithm, and we have 
\begin{eqnarray*}
| \ln (1+ \delta)| - (\delta - \frac{1}{2} \delta^2) | 
&=& \left |\sum_{i=3}^\infty (-1)^{i+1}\frac{\delta^{i}}{i} \right| 
= \left| \delta^2 \sum_{i=0}^{\infty} \frac{i+1}{i+3} \int_0^\delta (-x)^i dx \right|\\
&=& \left| \delta^2 \int_0^\delta \sum_{i=0}^{\infty} \frac{i+1}{i+3} (-x)^i dx \right|
\leq  \delta^2 \int_0^\delta \sum_{i=0}^{\infty} |x|^i dx \\
&\leq&  \delta^2 \int_0^\delta \frac{1}{1- |x|} dx
\leq  \delta^2 \int_0^\delta 2 dx = 2\delta^3,
\end{eqnarray*}
assuming $\delta \leq \frac12$.
We now show that in general $\delta$ is much smaller than $\frac12$.
Since we have used Algorithm \ref{alg:2krt} to compute $\hat{t}_p$ which is an approximation of $w_p^{\frac{1}{2^\ell}}$. Algorithm \ref{alg:2krt} uses  Algorithm \ref{alg:sqrt}.
Since the error bounds of Theorem \ref{thm1} and \ref{thm2} depend on the magnitude of $w_p$, the estimates of these theorems hold with $m=1$ because $w_p \in (1,2)$. We have 
$$|\hat{t}_p - w_p^{\frac{1}{2^\ell}}| \leq 2  \left( \frac{3}{4} \right)^{b-2}  \left( 2+ b +  \log_2 b  \right) \leq \frac18,$$
where we have set $b = \max\{5\ell,25\}$. Thus 
\begin{equation}
\label{eq:tp}
\hat{t}_p \leq w_p^{\frac{1}{2^\ell}} + 2  \left( \frac{3}{4} \right)^{b-2}  \left( 2+ b +  \log_2 b  \right) \leq  w_p^{\frac{1}{2^\ell}} + \frac18.
\end{equation}
Now we turn to $w_p^{\frac{1}{2^\ell}}$. Let $w_p = 1+ x_p$, $x_p \in (0,1)$. Consider the function $g(x_p) := (1+x_p)^{\frac{1}{2^\ell}}$. We take its Taylor expansion about $0$, and observing that
 $$g^{(j)}(x_p) = \left(\prod_{i=0}^{j-1} \frac{1-i2^\ell }{2^{\ell}}  \right)\; (1+x_p)^{\frac{1-j2^\ell}{2^\ell}} \leq \frac{(j-1)!}{2^\ell}, \;\; j \geq 1,$$
we have 
\begin{eqnarray*}
w_p^{\frac{1}{2^\ell}} - 1 = (1+x_p)^{\frac{1}{2^\ell}} - 1 &\leq& 
\frac{1}{2^\ell}\sum_{j=1}^\infty  \frac{(j-1)!}{j!} x_p ^j = \frac{1}{2^\ell}\sum_{j=1}^\infty  \frac1j x_p ^j\\
&=& \frac{1}{2^\ell} \sum_{j=1}^\infty \int_0^{x_p} s^{j-1} ds = 
\frac{1}{2^\ell} \int_0^{x_p}  \sum_{j=1}^\infty s^{j-1} ds \\
&\leq&  \frac{1}{2^\ell} \int_0^{x_p} \frac{1}{1-s}ds = \frac{1}{2^\ell} \ln{\frac{1}{1-x_p}}\\
&\leq& \frac{1}{2^\ell}\frac{n-m+p-1}{\log_2 e} \leq \frac{n-m+p-1}{2^\ell} , 
\end{eqnarray*}
because $x_p \leq 1 - 2^{-(n-m+p-1)}$, since $w$ has a fractional part of length $n-m$.
Observing that $n-m+p-1 \leq n$, and by setting $\ell \geq \lceil \log_2 8n \rceil$, we get 
\begin{equation}
\label{eq:wp}
w_p^{\frac{1}{2^\ell}} - 1 \leq \frac{n}{2^\ell} \leq \frac{n}{8n}\leq \frac18.
\end{equation}
Using equations (\ref{eq:tp}) and (\ref{eq:wp}), we get
\begin{equation}
\label{eq:tp2}
\hat{t}_p -1 \leq  \frac{1}{2^\ell} \frac{n}{\ln 2} + 2  \left( \frac{3}{4} \right)^{b-2}  \left( 2+ b +  \log_2 b  \right) \leq \frac14.
\end{equation}
Now we turn to the error of the algorithm. We have
\begin{eqnarray*}
 | z_p + \ln 2^{p -1} - \ln(w)| 
 &\leq&  | 2^{\ell} \hat{y}_p  - 2^{\ell} y_p | +  | 2^{\ell} y_p - 2^\ell \ln \hat{t}_p| +   | 2^\ell \ln \hat{t}_p- 2^\ell \ln w_p^{\frac{1}{2^\ell}}|\\
 &\leq & 2^{\ell - b}  + 2^{\ell}|2(\hat{t}_p -1)^3|   +  2^\ell |\hat{t}_p - w_p^{\frac{1}{2^\ell}}|\\
  &\leq & 2^{\ell - b}  
  + 2^{\ell}2\left( \frac{1}{2^\ell} \frac{n}{\ln 2} + 2  \left( \frac{3}{4} \right)^{b-2}  \left( 2+ b +  \log_2 b  \right) \right)^3\\ 
  &+&  2^\ell  \left(2  \left( \frac{3}{4} \right)^{b-2}  \left( 2+ b +  \log_2 b  \right)\right) 
\end{eqnarray*}
Since $b\geq \max\{5\ell,25\}$, we have
$\left( \frac{3}{4}\right)^{b/2} 2^\ell \leq 1$ and
$ \frac{2+b+\log_2 b}{\left( \frac43 \right)^{b/2}} \leq 1$. This yields
\begin{eqnarray}
\label{eq:lnErr1}
 | z_p + \ln 2^{p -1} - \ln(w)| 
&\leq& 
2^\ell \left( 2^{-b} +  \frac{2}{2^{3\ell}} \left(\frac{32}{9} + \frac{n}{\ln 2}   \right)^3 + \frac{32}{9} \left(\frac{3}{4}\right)^{b/2}    \right)   \nonumber \\
&\leq& 2^{-4\ell} + \frac{2}{2^{2\ell}}\left(\frac{32}{9} + \frac{n}{\ln 2}   \right)^3 + \frac{32}{9} \left(\frac{3}{4}\right)^{5\ell/2  \nonumber }\\
&\leq&  \left(\frac{3}{4}\right)^{5\ell/2} \left( 1+ \frac{32}{9} + 2\left(\frac{32}{9} + \frac{n}{\ln 2} \right)^3 \right)
  \end{eqnarray} 
Finally, from $b\geq 5\ell$ we have 
\begin{equation}
\label{eq:lnErr2}
(m-1) 2^{-b} \leq (m-1) 2^{-5\ell} \leq (m-1)\left( \frac{1}{4}\right)^{5\ell/2} \leq (m-1)\left( \frac{3}{4}\right)^{5\ell/2}
\end{equation}
Combining equations (\ref{eq:lnErrTot}), (\ref{eq:lnErr1}), and (\ref{eq:lnErr2}) then gives
\begin{eqnarray*}
 | z_p + (p -1)r - \ln(w)| &\leq&  \left(\frac{3}{4}\right)^{5\ell/2} \left( 1+ \frac{32}{9} + 2\left(\frac{32}{9} + \frac{n}{\ln 2} \right)^3 \right)+ (m-1)\left( \frac{3}{4}\right)^{5\ell/2}\\
  &\leq& \left(\frac{3}{4}\right)^{5\ell/2} \left( m+ \frac{32}{9} + 2\left(\frac{32}{9} + \frac{n}{\ln 2} \right)^3 \right),
\end{eqnarray*}
which is our desired bound.
\end{proof}

\begin{theorem} 
\label{thm4}
For $w>1$, given by $n$ bits of which the first $m$ correspond to its integer part, and $1\geq f \geq0$ given by $n_f$ bits of accuracy, Algorithm \ref{alg:fracPower} computes an approximation $\hat{z}$ of $w^{f}$ such that
\begin{equation}
|\hat{z} - w^{f}|  \leq \left( \frac{1}{2} \right)^{\ell  - 1  } , 
\end{equation}
where $\ell \in \nat$ is a parameter specified in the algorithm that is used to determine the number 
 $b =\max \{ n,n_f, \lceil 5(\ell + 2m +\ln n_f)\rceil, 40 \}$
of bits after the decimal point in which arithmetic will be performed, and therefore it determines the error.
Algorithm \ref{alg:fracPower} uses Algorithm \ref{alg:2krt} which computes power of 2 roots of a given number. 
\end{theorem}

\begin{proof}
First observe that the algorithm is exact for the cases $f=1$ or $f=0$. Therefore, without loss of generality assume $0<f<1$.

Consider the $n_f$ bit number $f$ and write its binary digits $f_i \in \{0,1\}$ explicitly as $f=0.f_1f_2...f_{n_f} = \sum_{i=1}^{n_f} f_i /2^i$. Denote the set of non-zero digits $\p: =  \{ 1\leq i \leq n_f : f_i \neq 0 \} $, $p:= |\p|$ with $1\leq p \leq n_f$.  
 Observe 
$$w^{f} = w^{0.f_1f_2...f_{n_f}} = w^{ \sum_{i=1}^{n_f} f_i /2^i} = \prod_{i=1}^{n_f} w_i^{f_i}  = \prod_{i\in \p} w_i, $$  
where again $w_i := w^\frac{1}{2^i}$. Let $\hat{w}_i \simeq w^\frac{1}{2^i}$ denote the outputs of Algorithm \ref{alg:2krt}. 
We have 
\begin{equation}    \label{fracPowTotErr}
\left|\hat{z} - w^{f}\right|  \leq \left|\hat{z} - \prod_{i\in \p} \hat{w}_i \right|  + \left|   \prod_{i\in \p} \hat{w}_i -  \prod_{i\in \p} w_i \right| ,
\end{equation}
where these terms give bounds for the repeated multiplication error, and the error from the $2^i$th roots as computed by Algorithm \ref{alg:2krt}, respectively.

Consider the second term. Partition $\p$ disjointly into two sets as $\p_1 : =  \{ i\in \p: \hat{w}_i  \geq w_i \} $ and  $\p_2 : =  \{ i\in \p: \hat{w}_i<w_i \} $. Observe that, for any $i$, from equation (\ref{eq:thm2}) of Theorem \ref{thm2} we have $|\hat{w}_i - w_i| \leq 2  \left( \frac{3}{4} \right)^{b-2m}  \left( 2+ b +  \log_2 b  \right) =: \e$. Observe $w_i, \hat{w}_i \geq 1$. 
First assume $\prod_{i\in \p} \hat{w}_i \geq \prod_{i\in \p} w_i $. Then
\begin{eqnarray*}
\left|   \prod_{i\in \p} \hat{w}_i -  \prod_{i\in \p} w_i \right|  
&=&  \prod_{i\in \p_1} \hat{w}_i \prod_{j\in \p_2} \hat{w}_j  -  \prod_{i\in \p_1} w_i  \prod_{j\in \p_2} w_j \\
&\leq&  \prod_{i\in \p_1} (w_i + \e) \prod_{j\in \p_2} w_j  -  \prod_{i\in \p_1} w_i  \prod_{j\in \p_2} w_j\\
&\leq&  \prod_{i\in \p_1} w_i  \prod_{j\in \p_2} w_j \left(  \prod_{k\in \p_1} \left(    1 + \frac{\e}{w_k}  \right)  - 1    \right)\\
&\leq&  w^f  \left(  \left(1+\e)^p -1 \right)\right) \\
&\leq&  w^f  \left( e^{p\e} -1 \right) 
\end{eqnarray*}
Conversely, assume $\prod_{i\in \p} \hat{w}_i < \prod_{i\in \p} w_i $. Then similarly we have
\begin{eqnarray*}
\left|   \prod_{i\in \p} \hat{w}_i -  \prod_{i\in \p} w_i \right|  
&=&   \prod_{i\in \p_1} w_i  \prod_{j\in \p_2} w_j - \prod_{i\in \p_1} \hat{w}_i \prod_{j\in \p_2} \hat{w}_j  \\
&\leq&   \prod_{i\in \p_1} w_i  \prod_{j\in \p_2} w_j - \prod_{i\in \p_1} w_i \prod_{j\in \p_2} (w_j-\e) \\
&\leq&  \prod_{i\in \p_1} w_i  \prod_{j\in \p_2} w_j \left( 1 - \prod_{k\in \p_2} \left(    1 - \frac{\e}{w_k}  \right)     \right)\\
&\leq&  w^f  \left( 1- \left(1-\e)^p \right)\right) \\
&\leq&  w^f  \left( 1- \left(  2 - e^{p\e}   \right)\right) \\
&\leq&  w^f  \left( e^{p\e} -1 \right) ,
\end{eqnarray*}
where we have used the inequality $\left(1-\e\right)^p -1 \geq 1- e^{pe}$.\footnote{This inequality follows trivially from term by term comparison of the binomial expansion of the left hand side with the Taylor expansion of the right hand side.}
So conclude that $ \left|   \prod_{i\in \p} \hat{w}_i -  \prod_{i\in \p} w_i \right|   \leq  w^f  \left( e^{p\e} -1 \right)$ always.
Furthermore, for $a \geq 0$ we have
$$ e^a -1 = a + a^2/2! + a^3/3!+ \dots = a(1+ a/2! + a^2/3! +\dots ) \leq a(1+a + a^2/2! + \dots)  = ae^a $$
which yields
\begin{equation}
\left|   \prod_{i\in \p} \hat{w}_i -  \prod_{i\in \p} w_i \right| \leq  w^{f} \;p\e \; e^{p\e} 
\end{equation}

Next consider the error resulting from truncation to $b$ bits of accuracy in the products computed in step 13 of the algorithm. 
For each multiplication, we have $\hat{z} = z +\xi$, with error $|\xi| \leq 2^{-b}$. For notational simplicity, reindex the set $\p$ as $\{1,2,\dots,p\}$ so that $\prod_{i\in \p} \hat{w}_i = \prod_{i=1}^p \hat{w}_i$. 
Let $z_i = \hat{w}_1\hat{w}_2\dots\hat{w}_i$, $i=1,2,\dots,p$ be the exact products, and 
let the approximate products be $\hat{z}_i = \hat{z}_{i-1}\hat{w_i}+\xi_i$, $i=2,\dots,p$, $\hat{z}_1 = 1$.
We have 
\begin{eqnarray*}
\hat{z}_1 &=& \hat{w}_1  \\
\hat{z}_2 &=& \hat{z}_1 \hat{w}_2 + \xi_2  \\
&\vdots& \\
\hat{z}_p &=& \hat{z} _{p-1}\hat{w}_p + \xi_p. 
\end{eqnarray*}
 Then we have
\begin{eqnarray*}
 \left|\hat{z}_p - \prod_{i\in \p} \hat{w}_i \right|  &=& | \hat{z}_p - z_p  | 
=  |\hat{z}_{p-1}\hat{w}_p + \xi_p - z_{p-1}\hat{w}_p|
 \leq  |\hat{z}_{p-1}\hat{w}_p - z_{p-1}\hat{w}_p| + |\xi_p|\\
  &\leq&  |\hat{z}_{p-1} - z_{p-1}|\hat{w}_p  + |\xi_p|\\
    &\leq& \left( |\hat{z}_{p-2} - z_{p-2}|\hat{w}_{p-1}  + |\xi_{p-1}| \right)\hat{w}_p  + |\xi_p|\\
        &\leq&  |\hat{z}_{p-2} - z_{p-2}|\hat{w}_{p-1}\hat{w}_p   + |\xi_{p-1}|\hat{w}_p  + |\xi_p|\\
&\vdots& \\
&\leq&  |\hat{z}_{1} - z_{1}|\hat{w}_2 \hat{w}_3 \hat{w}_4\dots \hat{w}_{p-1}\hat{w}_p  +  |\xi_{2}|\hat{w}_3 \hat{w}_4\dots \hat{w}_{p-1}\hat{w}_p    +   \dots + |\xi_{p-1}|\hat{w}_p  + |\xi_p|\\
&\leq&  2^{-b} \left( \hat{w}_3 \hat{w}_{4}\dots \hat{w}_{p-1} \hat{w}_p+ \hat{w}_{4}\dots \hat{w}_{p-1} \hat{w}_p+ \dots + \hat{w}_p \hat{w}_{p-1} +\hat{ w}_p +1 \right)\\
&\leq&  2^{-b} \left(p-1\right) w^\frac14
\leq  2^{-b} n_f w^f \leq  2^{-b} n_f w
\end{eqnarray*}
where the last line follows from observing each of the $p-1$ terms in the sum is less than $w_2 = w^\frac14$.

Thus, equation (\ref{fracPowTotErr}) yields total error
\begin{eqnarray*}
\left|\hat{z} - w^{f}\right| &\leq& 
\left|\hat{z} - \prod_{i\in \p} \hat{w}_i \right|  
+ \left|   \prod_{i\in \p} \hat{w}_i -  \prod_{i\in \p} w_i \right|  \\
&\leq& 
  2^{-b} \left(p-1\right) w^\frac14
+   w^{f} \:p\e \: e^{p\e}  \\
&\leq& 
  2^{-b} \; n_f w
+   w \:n_f \e \: e^{n_f\e} 
\end{eqnarray*}
Furthermore, using $\e = 2\left( \frac{3}{4} \right)^{b-2m}  \left( 2+ b +  \log_2 b  \right)$, $1\leq w\leq 2^m$, and as we have chosen $b$ sufficiently large such that $n_f \e \leq 1$ (to be shown later), we have
\begin{eqnarray*}
\left|\hat{z} - w^{f}\right| 
&\leq&  2^{-b} \; n_f w
+   w  \:n_f   \;2 \left( \frac{3}{4} \right)^{b-2m}   \left( 2+ b +   \log_2 b \right) e \\
&\leq&  n_f 2^{m-b}    +  n_f 2^{m} \left( \frac{3}{4} \right)^{b-2m}  2e\left( 2+ b +  \log_2 b  \right)  \\
\end{eqnarray*}
We have selected $b\geq \max\{n, n_f, \lceil 5(\ell + 2m + \log_2 n_f)\rceil , 40\}$ such that several inequalities are satisfied. 
From $b \geq \lceil 5(\ell + 2m + \log_2 n_f)\rceil  \geq  \ell + m + \log_2 n_f$, it follows that $n_f 2^{m-b} \leq 2^{-\ell}$.
Furthermore, for $b \geq 40$ , we have $2e \left( 2+ b +  \log_2 b  \right)  \leq (\frac{4}{3})^{b/2}$. 
Finally, $b \geq 5\left(    \ell + 2m + \log_2 n_f    \right) \geq \left(  \frac{2}{\log_2 4/3} \right)  \left(   \log_2 n_f + \ell + 2m     \right)$ implies that $n_f 2^{2m}  \left( \frac{3}{4} \right)^{b/2} \leq 2^{-\ell}$. Plugging these inequalities into the previous equation yields
\begin{eqnarray*}
\left|\hat{z} - w^{f}\right| 
&\leq&  n_f 2^{m} \frac{1}{2^b}  
  +  n_f 2^{m} \left( \frac{4}{3} \right)^{2m}   \left( \frac{3}{4} \right)^{b/2}    \left(  \left( \frac{3}{4} \right)^{b/2}  2e\left( 2+ b +  \log_2 b  \right)  \right) \\
&\leq&  \left( \frac{1}{2} \right)^{\ell}  
 +  n_f 2^{m} \left( \left(\frac{4}{3} \right)^{2}\right)^m   \left( \frac{3}{4} \right)^{b/2}   \\
&\leq&  \left( \frac{1}{2} \right)^{\ell}  
 +  n_f 2^{2m}  \left( \frac{3}{4} \right)^{b/2}    \\
&\leq& 2  \left( \frac{1}{2} \right)^{\ell}  = \left( \frac{1}{2} \right)^{\ell  -1  }  
\end{eqnarray*}
as was to be shown.

Finally, for completeness, from the above inequalities and $b \geq 5\left(   \log_2 n_f + \ell + 2m     \right) 
\geq   \frac{2}{\log_2 4/3}  \left(  \log_2 n_f  - \log_2 e \right) + 4m $,
 we have
\begin{eqnarray*}
n_f \e &=& n_f 2\left( \frac{3}{4} \right)^{b-2m}  \left( 2+ b +  \log_2 b  \right) 
\leq  n_f \left( \frac{3}{4} \right)^{b/2-2m}  \left( \frac{1}{e}  \right) \\
&\leq&  n_f \left( \frac{3}{4} \right)^{   \log_{4/3} n_f/e   }  \left( \frac{1}{e}  \right) \leq 1 \\
\end{eqnarray*}

\end{proof}

\begin{cor} 
\label{cor1}
Let $w>1$ and $1\geq f \geq0$ as in Theorem \ref{alg:fracPower} above. Suppose $f$ is an approximation of a number $1\geq F \geq 0$ accurate to $n_f$ bits. Then Algorithm \ref{alg:fracPower}, with $f$ as input, computes an approximation $\hat{z}$ of $w^{F}$ such that
\begin{equation}
|\hat{z} - w^{F}|  \leq     \left( \frac{1}{2} \right)^{\ell  - 1  }  + \frac{w\ln w}{2^{n_f}}, 
\end{equation}
\end{cor} 

\begin{proof}
Consider the error from the approximation of the exponent $F$ by $f$. We have $|F- f| \leq 2^{-n_f}$. Let $g(f):= w^f$. Then $g'(f) = w^f \ln w$. By the mean value theorem, we have
\begin{eqnarray*}
\left|w^{f} - w^{F} \right|  \leq \sup_{f\in (0,1)} g'(f) |F - f| \leq 2^{-n_f} \: w \ln w,
\end{eqnarray*}
which gives
$$|\hat{z} - w^{F}|  \leq |\hat{z} - w^{f}| + |w^{f} - w^{F} | \leq   \left( \frac{1}{2} \right)^{\ell  - 1  }  + \frac{w\ln w}{2^{n_f}}.$$
\end{proof}

\begin{theorem} 
\label{thm5}
For $0\leq w < 1$, represented by $n$ bits of which the first $m$ correspond to its integer part, and $1\geq f \geq0$ given by $n_f$ bits of accuracy, Algorithm \ref{alg:fracPower2} computes an approximation $\hat{t}$ of $w^{f}$ such that
\begin{equation}
|\hat{t} - w^{f}| \leq   \frac{1}{2^{\ell - 3 }}
\end{equation}
where $\ell \in \nat$ is a parameter specified in the algorithm that is used to determine the number $b = \max \{ n,n_f, \lceil 2\ell + 6m +2\ln n_f \rceil,40\}$ of bits after the decimal point in which arithmetic will be performed, and therefore also will determine the error.
Algorithm \ref{alg:fracPower2} uses Algorithm \ref{alg:fracPower}, which computes $w^f$ for the case $w \geq 1$, and 
also 
Algorithm \ref{alg:inv} which computes the reciprocal of a number  $w \geq 1$.
\end{theorem}

\begin{proof}
First observe that the algorithm is exact for the cases $f=1$, $f=0$, or $w=0$. 
Therefore, without loss of generality assume $0<f<1$ and $0<w<1$. 

Let all variables be defined as in Algorithm \ref{alg:fracPower2}. We shall first consider the error of each variable and use this to bound the overall error of algorithm.

Firstly, the input $0<w<1$ is rescaled  to $x:=2^k w \geq 1 > 2^{k-1}w$ exactly, by $k$-bit left shift,where $k$ is a positive integer. An example circuit for computing $k$ is given in Fig. \ref{fig-ShiftInteger} above. Observe that we have
$$ w^f = x^f/2^{kf} = x^f/ ( 2^{\lfloor kf  \rfloor}2^{  \{ kf \}} ).$$
We also have $\log_2 \frac{1}{w} \leq k < \log_2 \frac{1}{w}+1$.

The product $c=kf  < k \leq n-m$ is computed exactly in fixed precision arithmetic because the number of bits after the decimal point in $kf$ is at most
 $n_f \leq b$, where $b$ is the number of bits in which arithmetic is performed.

Next consider $\hat{z}=$ FractionalPower($x$, $f$, $n$, $m$, $n_f$, $\ell$), which approximates $z = x^f$. From Theorem \ref{thm4} we have 
$ e_z: = |\hat{z} - z| \leq \frac{1}{2^{\ell -1}}$.
Similarly, $\hat{y}= $ FractionalPower($2$, ${\{ c \}} $, $n$, $m$, $n_f$, $\ell$) approximates $y=2^{\{ c \}}$, with $e_y:= |\hat{y} - y| \leq \frac{1}{2^{\ell - 1}}$.
Furthermore, for $\hat{s} = $ INV($\hat{y},n,1,2\ell$) which approximates $s =1/\hat{y}$, from Corollary \ref{cor0} we have 
$e_s :=|\hat{s} - s| \leq \frac{2+\log_2\ell}{2^{2\ell}}$, which satisfies $e_s \leq \frac{1}{2^\ell}$ for $\ell \geq 2$.

Finally, observe $2^{- \lfloor c \rfloor}  \hat{z}$ is computed exactly by a right shift of  $\hat{z}$. This is used to compute $t = 2^{- \lfloor c \rfloor}  \hat{z} \hat{s}$, which is again truncated to $b$ decimal bits to give $\hat{t}$ with $e_t:= |\hat{t} - t| \leq 2^{-b}$, and returned.

Now we turn to the total error of our algorithm. By our variable definitions, $w^f = 2^{- \lfloor c \rfloor}  z /y $. We have
\begin{eqnarray*}
| \hat{t} - w^f |  &\leq&  |  \hat{t} - t | 
+ | 2^{- \lfloor c \rfloor}  \hat{z}\hat{s} - 2^{- \lfloor c \rfloor}  z\hat{s}  |
+ |2^{- \lfloor c \rfloor}  z\hat{s}   - 2^{- \lfloor c \rfloor}  z s   |
+|2^{- \lfloor c \rfloor}  z s -   2^{- \lfloor c \rfloor}  \frac{z}{y}|\\
&\leq& 2^{-b}
+    2^{- \lfloor c \rfloor} \hat{s} | \hat{z} -  z  |
+ 2^{- \lfloor c \rfloor}  z  |\hat{s}   -  s   |
+ 2^{- \lfloor c \rfloor}  z | s -    \frac{1}{y}|\\
&\leq& 2^{-b}
+   \frac{1}{ 2^{\lfloor c \rfloor} } \frac{1}{2^{\ell -1}}
+ w^f 2^{  \{ c \}}  |\hat{s}   -  s   |
+ \frac{x^f}{2^{ \lfloor c \rfloor} } \frac{1}{2^{ \{ c \} } }|  \frac{2^{ \{ c \} } }{\hat{y}} -  1|,\\
\end{eqnarray*}
where we have used $\hat{s} \leq s = \frac{1}{\hat{y}} \leq 1$ because as remarked in the proof of Theorem \ref{thm1}, the algorithm computing the reciprocal underestimates it value, i.e. $\frac{1}{\hat{y}} \leq \frac{1}{y} = 2^{-\lfloor kf\rfloor } \leq 1$. 
Observe we have $w^f = \frac{x^f}{2^{ \lfloor c \rfloor} } \frac{1}{2^{ \{ c \} } } \leq 1$. Moreover, $2^{\{ kf\}} \leq 2$. Hence, as shown in the proof of Theorem \ref{thm1}, that $\hat{y} \geq 1$.  This, together with the error bounds of Theorem \ref{thm4} and of 
Corollary \ref{cor0} yields
\begin{eqnarray*}
| \hat{t} - w^f |   
&\leq& 
2^{-b} + \frac{1}{2^{\ell -1}}
+ 2 |\hat{s}   -  s   |
+  \frac{1}{\hat{y}}  |  2^{ \{ c \}  } -\hat{y}   |\\
 &\leq& 2^{-b} +  \frac{1}{2^{\ell -1}}
+\frac{2}{2^{\ell}}
+\frac{1}{2^{\ell -1}}\\
 &\leq& 4 \frac{1}{2^{\ell -1 }} = \frac{1}{2^{\ell - 3 }}
\end{eqnarray*}
\end{proof}

\begin{cor} 
\label{cor2}
Let $0 \leq w < 1$ and $1\geq f \geq 0$ as in Theorem \ref{alg:fracPower2} above. Suppose $f$ is an approximation of a number $1\geq F \geq 0$ accurate to $n_f$ bits. Then Algorithm \ref{alg:fracPower2}, with $f$ as input, computes an approximation $\hat{t}$ of $w^{F}$ such that
\begin{equation}
|\hat{t} - w^{F}|  \leq     \left( \frac{1}{2} \right)^{\ell  - 2  }  + \frac{w\ln w}{2^{n_f}}.
\end{equation}
\end{cor} 

\begin{proof}
The proof is similar to that of Corollary \ref{cor1}.
\end{proof}

\addcontentsline{toc}{section}{References} 
\bibliographystyle{plain}
\bibliography{bib}

\end{document}